\documentclass[a4paper,reqno,12pt]{amsart}

\usepackage{amsmath,amsthm,amssymb,color,graphicx}
\usepackage{amscd,verbatim}
\usepackage[all]{xy}

\usepackage{hyperref}

\usepackage{fullpage,enumerate}

\newcommand{\bZ}{{\mathbb Z}}

\renewcommand\wp{\widehat{\varphi}}
\newcommand\wf{\widehat{f}}
\newcommand\rp{\textup{)}}
\newcommand\lp{\textup{(}}
\newcommand{\Maps}{\operatorname{Maps}}
\newcommand\bE{\mathbf E}

\newcommand\bh{\mathbf h}

\newcommand\co{\colon\,}
\newcommand{\beq}{\begin{equation}}
\newcommand{\eeq}{\end{equation}}
\newcommand{\bea}{\begin{eqnarray}}
\newcommand{\eea}{\end{eqnarray}}
\renewcommand{\H}{{\operatorname{H}}}
\def\K{ \hbox{\rm K}}
\newcommand{\Z}{\ensuremath{\mathbb Z}}
\newcommand{\R}{\ensuremath{\mathbb R}}
\def\im{ \hbox{\rm Im}}
\def\ker{ \hbox{\rm Ker}}
\def\hsp{,\hspace{.7cm}}

\theoremstyle{plain}
\newtheorem{theorem}{Theorem}
\newtheorem{lemma}[theorem]{Lemma}

\theoremstyle{definition}

\theoremstyle{remark}

\numberwithin{equation}{section}
\numberwithin{theorem}{section}
\numberwithin{figure}{section}
\numberwithin{table}{section}

\newcommand{\cE}{{\mathcal E}}

\newcommand{\cK}{{\mathcal K}}
\newcommand{\cL}{{\mathcal L}}
\newcommand{\cO}{{\mathcal O}}

\newcommand{\CC}{{\mathbb C}}

\newcommand{\QQ}{{\mathbb Q}}
\newcommand{\RR}{{\mathbb R}}
\newcommand{\TT}{{\mathbb T}}
\newcommand{\ZZ}{{\mathbb Z}}

\newcommand{\HH}{{\mathbb H}}

\newcommand{\bF}{{\mathbf F}}

\newcommand{\CS}{\text{CS}}
\newcommand{\cs}{\text{cs}}

\newcommand{\fg}{{\mathfrak g}}

\newcommand{\sfG}{{\mathsf G}}

\newcommand{\sfE}{{\mathsf E}}

\newcommand{\sfU}{{\mathsf U}}
\newcommand{\sfSU}{{\mathsf{SU}}}

\renewcommand{\mp}{\mathcal{P}}

\begin{document}

\title[Spherical T-Duality]
{Spherical T-Duality}

\author[P Bouwknegt]{Peter Bouwknegt}

\address[Peter Bouwknegt]{
Mathematical Sciences Institute, and
Department of Theoretical Physics,
Research School of Physics and Engineering, 
The Australian National University, 
Canberra, ACT 0200, Australia}

\email{peter.bouwknegt@anu.edu.au}

\author[J. Evslin]{Jarah Evslin}

\address[Jarah Evslin]{
High Energy Nuclear Physics Group, 
Institute of Modern Physics,
Chinese Academy of Sciences,
Lanzhou, China}

\email{jarah@impcas.ac.cn}

\author[V Mathai]{Varghese Mathai}

\address[Varghese Mathai]{
Department of Pure Mathematics,
School of  Mathematical Scienes, 
University of Adelaide, 
Adelaide, SA 5005, 
Australia}

\email{mathai.varghese@adelaide.edu.au}

\begin{abstract}
We introduce {\it{spherical}} T-duality, which relates pairs of the form $(P,H)$ consisting of a principal $\sfSU(2)$-bundle $P\rightarrow M$ and a 7-cocycle $H$ on $P$.  Intuitively spherical T-duality exchanges $H$ with the second Chern class $c_2(P)$. 
Unless $dim(M)\leq 4$, not all pairs admit spherical T-duals and the spherical T-duals are not always unique.  Nonetheless, we prove that all spherical T-dualities induce a degree-shifting isomorphism on the 7-twisted cohomologies of the bundles and, when $dim(M)\leq 7$, also their integral twisted cohomologies and,  when $dim(M)\leq 4$, even their 7-twisted K-theories.   While spherical T-duality does not appear to  relate equivalent string theories, it does provide an identification between conserved charges in certain distinct IIB supergravity and string compactifications. 
\end{abstract}

\maketitle


\vspace{-.2in}

\section{Introduction}

In earlier papers \cite{BEM, BEM2, BHM, BHM05}, we showed that to each pair $(P,H)$ of a manifold $P$ with a free circle action and integral 3-cocycle $H$ on $P$, one can uniquely associate a T-dual pair $(\widehat P,\widehat H)$ of a manifold $\widehat P$ with a free circle action and a cocycle $\widehat H$ on $\widehat P$.  While the space of orbits of the two circle actions are the same, $P$ and $\widehat P$ are in general not homeomorphic.  Nonetheless we showed that T-duality induces a number of degree-shifting isomorphisms between various structures such as twisted cohomology and twisted K-theory on $P$ and $\widehat P$.  Later \cite{gualtieri,cavalcanti} it was shown that T-duality also induces isomorphisms on Dirac structures, Courant algebroids,  generalized complex structures and generalized K\"ahler structures.

The free circle action on $P$ gives it the structure of a principal $\sfU(1)$-bundle $P\rightarrow M$.  One may associate a complex line bundle to this $\sfU(1)$-bundle, such that the $\sfU(1)$-bundle is just the sphere $S^1\subset\CC$ subbundle.   In this note we will answer the following question: \\ \noindent{\it Just how much of this structure carries over to the case of $S^3\subset \HH$ subbundles of quaternionic line bundles?}   

There are several ways to generalize the definition of the pair $(P,H)$ and the answer to the question depends upon this choice.  We will restrict our attention to the simplest choice, in which $P$ is a principal $\sfSU(2)$-bundle over $M$ and $H$ is a 7-cocycle on $P$
\begin{equation}\label{circle1}
\begin{CD}
\sfSU(2) @>>> \,  P \\
&& @V \pi VV \\
&& M \end{CD}
\end{equation}

When $dim(M)\leq 4$ we will find that things work essentially identically to the circle bundle case. 
In this case principal $\sfSU(2)$-bundles over a compact oriented four dimensional manifold $M$ are classified 
by $\H^4(M;\ZZ) \cong \ZZ$ via the 2nd Chern class $c_2(P)$. This can be seen using the well known isomorphism,
$\H^4(M;\ZZ) \cong [M, S^4]\cong \ZZ$ and noting that there is a canonical principal $\sfSU(2)$-bundle 
$Q\to S^4$, known as the Hopf bundle, whose 2nd Chern class is the generator of $\H^4(S^4;\ZZ)\cong \ZZ$.  The orientation of $M$ and $\sfSU(2)$ imply that $\pi_*$ is a canonical isomorphism $\H^7(P;\ZZ)\cong\H^4(M;\ZZ)\cong\Z$.  The dual bundle $\widehat\pi:\widehat P\rightarrow M$ is defined by $c_2(\widehat P)=\pi_*H$ while the dual 7-cocycle $\widehat H\in\H^7(\widehat P)$ is related to $c_2(P)$ by the isomorphism $\widehat\pi_*$.  We will see that this spherical T-duality map induces degree-shifting isomorphisms between the real and integral twisted cohomologies of $P$ and $\widehat P$ and also between the 7-twisted K-theories.

Beyond dimension 4 the situation becomes more complicated as not all integral 4-cocycles of $M$ are realized as $c_2$ of a principal 
$\sfSU(2)$-bundle $\pi:P\rightarrow M$ and multiple bundles can have the same $c_2(P)$.  
Proposition 3.6 in Granja's thesis \cite{Granja}  states a sufficient condition for a cohomology class in $\H^4(M;\ZZ)$ to be the 2nd 
Chern class of a principal $\sfSU(2)$-bundle over any manifold $M$ of dimension $\leq d$.  He shows that there exists a positive integer 
$N(d)$ (depending only on the dimension of $M$) such that any class in $N(d)\times \H^4(M;\ZZ)$ is the 2nd Chern class of a 
principal $\sfSU(2)$-bundle over $M$. However this principal $\sfSU(2)$-bundle is not in general unique when $d>4$.
Note that $N(4)=1$.
We will simply assert that the T-dual $\widehat\pi:\widehat P\rightarrow M$ be any principal $\sfSU(2)$-bundle with 
$c_2(\widehat P)=\pi_* H$, and with $\widehat H$ defined such that $\widehat\pi_*\widehat H=c_2(P)$ 
with $\widehat p^*H=p^*\widehat H$ on the correspondence space $P\times_{M}  \widehat P $, which is defined by the commutative diagram,
\begin{equation} \label{correspondenceb}
\xymatrix 
@=4pc @ur 
{ P \ar[d]_{\pi} & 
P\times_{M}  \widehat P \ar[d]^{p=\pi\otimes 1} \ar[l]_{\widehat p=1\otimes \widehat\pi} \\ {M} & \widehat P\ar[l]^{\widehat \pi}}
\end{equation}
When $dim(M)\leq 6$ this condition specifies $\widehat H$ uniquely. 

Thus far it may seem as though the generalization of T-duality to $\sfSU(2)$ bundles fails to be unique and so has no applications.  The reason that spherical T-duality is interesting is that, as is proved in Section 5, whenever a pair 
$(P,H)$ does admit a T-dual, in the sense that there is an $\sfSU(2)$-bundle $\widehat P\rightarrow M$ with $c_2(\widehat P)=\pi_*H$, 
then every such T-dual induces an isomorphism of the $d_H=d-H\wedge$ twisted cohomology of $P$ with the $d_{\widehat H}$ 
twisted cohomology of $\widehat P$ with a shifted degree.   Furthermore, as is shown in Section 6, when $dim(M)\leq 7$ it also 
induces an isomorphism in integral twisted cohomology and when $dim(M)\leq 4$ this isomorphism even lifts to integral twisted K-theory.

In Section 3, we construct a pair of classifying spaces $(R, S)$, where $R$ consists of (equivalence classes of)
pairs $(P, H)$ over $M$ consisting of a principal
$\sfSU(2)$-bundle $P \to M$ together with a class $H \in \H^7(P;\ZZ)$, 
and $S$ consists of (equivalence classes of) spherical T-dual pairs of such pairs. The problem with spherical T-duality
in higher dimensions is encapsulated by the observation that $R\ne S$, as a result of the fact that $\sfSU(2)$ is not a model for $K(\Z,3)$. More precisely, there is a map $g:S\rightarrow R$ and a pair $(P,H)$ corresponds to a map $f:M\rightarrow R$.    T-duals of $(P,H)$ correspond to lifts $\tilde f:M\rightarrow S$ such that $g\tilde f=f$.  Using rational homotopy theory,
we observe that 
the rationalizations $R_\QQ = S_\QQ$ are equal and so spherical T-duality works nicely over the rationals.

Sections 7 and 8 relate our results to String Theory.
In Section 7, we argue that the 7-twisted cohomologies featured in our main theorems classify certain conserved charges in 
type IIB supergravity.  We conclude that spherical T-duality provides a one to one map between conserved charges in certain 
topologically distinct compactifications and also a novel electromagnetic duality on the fluxes.
In Section 8, we suggest that spherical T-duality preserves the spectra of certain spherical 3-branes that wrap 
$S^3$ cycles in some spacetime $X$, i.e.\ by replacing closed strings, described by $\text{Maps}(S^1,X)$, by 
spherical 3-branes
(or `closed quaternionic strings') described by $\text{Maps}(S^3,X) = \text{Maps}(S(\HH),X)$, where 
$S(\HH)$ denotes the unit sphere in the quaternions $\HH$.

Section 9 contains speculations and open questions, such as whether the missing spherical T-duals 
in higher dimensions can be obtained via noncommutative geometry; the higher rank case of
principal $\sfSU(2)^r$-bundles $P$ with flux $H\in \H^7(P; \ZZ)$, and the quest for higher twisted Courant algebroids.
 
While this work was at an early stage we learned \cite{hishamdisc} of independent work on spherical T-duality which will appear in Ref.~\cite{hishamwest}.

\smallskip

\noindent{\bf Acknowledgements.} The authors are grateful to Diarmuid Crowley for useful information 
regarding the classification of principal $\sfSU(2)$-bundles. Hisham Sati has  informed the authors 
that he had independently considered studying a similar duality.
JE is supported by NSFC MianShang grant 11375201 and a fellowship at the Australian National University. PB and VM thank the 
Australian Research Council for support via ARC Discovery Project grants 
DP110100072 and DP130103924.

\vspace{-.1in}
\tableofcontents

\section{Construction of the spherical T-Dual: Gysin Sequence Approach} \label{gyssez}

\subsection{The Gysin complex}

In this section we will motivate the existence of a ``spherical T-dual'' (in a limited sense, to be explained later) for
principal $\sfSU(2)$-bundles $\pi:P\to M$.  Throughout we will take $M$ to be a compact, oriented, manifold and
we identify $\sfSU(2) = S^3$.  First we recall
\begin{theorem}
Let $\pi:P\to M$ be a principal $\sfSU(2)$-bundle.  We have the following exact sequence, known as the Gysin sequence
of \v Cech cohomology groups over the integers
\begin{equation} \label{eqn:eqBAa}
\xymatrix{
\cdots \ar[r] & \H^k(M) \ar[r]^{\pi^*} & \H^k(P) \ar[r]^{\pi_*} & \H^{k-3}(M) \ar[r]^{c_2 \cup} & \H^{k+1}(M) \ar[r] & \cdots}
\end{equation}
where $\pi^*$ denotes the pull-back map, $\pi_*$ the push-forward map and $c_2\cup$ the cup product with the 
2nd Chern class of  $c_2(P)\in \H^4(M,\ZZ)$ of $P$.
Here we have  identified the Euler class of the $S^3$-bundle with the 2nd Chern class of the associated 
vector bundle $E = P\times_{S^3} \RR^4$ (or, equivalently, of the associated quaternionic line bundle $L=
P\times_{S^3} \HH$, where $S^3$ acts on $\HH$ through multiplication of unit quaternions).
We have a similar Gysin sequence in de Rham cohomology, which will also be used in later sections.
\end{theorem}

The Gysin sequence suggests that we should look at pairs $(P,H)$, where $\pi: P\to M$ is a principal 
$\sfSU(2)$-bundle, and $H\in \H^7(P,\ZZ)$.  We can then take $\pi_* H \in \H^4(M,\ZZ)$, and the question arises
whether $\pi_* H$ is the 2nd Chern class of some (isomorphism class of) `spherical T-dual'  principal $\sfSU(2)$-bundle $\widehat P$. 
While in the case of principal $\sfU(1)$-bundles, we have an isomorphism $[M,B\sfU(1)] \cong \H^2(M,\ZZ)$,
unfortunately, isomorphism classes of principal $\sfSU(2)$-bundles over $M$ are not completely classified by $\H^4(M,\ZZ)$.
We do have a map $[M,B\sfSU(2)] \rightarrow \H^4(M,\ZZ)$ but in general this map can both fail to be injective and surjective.
We will come back to this point in more detail, with examples, in later sections.  For the moment, assume that there
exists a dual principal $\sfSU(2)$-bundle $\widehat P$ such that  $c_2(\widehat P) = \pi_* H$ (from 
the remark above, this bundle need not be unique).
The Gysin sequence for $\widehat \pi : \widehat P\to M$, then implies that there exists a $\widehat H \in \H^7(\widehat P,\ZZ)$
such that $\widehat \pi_* \widehat H = c_2(P)$, and that $\widehat H$ is determined by this condition up to an element $\widehat \pi^* h$,
with $h\in \H^7(M,\ZZ)$.  As in the circle bundle case we aim to fix the non-uniqueness in $\widehat H$ by imposing the 
condition $\widehat p^*H - p^* \widehat H = 0\in \H^7(P\times_M \widehat P,\ZZ)$ on the correspondence space $P\times_M \widehat P$ (see Figure below)
\begin{equation} \label{eqn:eqBAb}
\xymatrix{
& P\times_M\widehat P \ar[dl]_{\widehat  p= 1\otimes \widehat \pi}  \ar[dr]^{p = \pi\otimes 1} \\
P \ar[dr]_{\pi} && \widehat P \ar[dl]^{\widehat \pi}  \\
& M & 
}
\end{equation}
We have
\begin{theorem} \label{thm:thBAb}
Let $P$ be a principal $\sfSU(2)$-bundle with 2nd Chern class $c_2 \equiv c_2(P) \in \H^4(M)$,
and let  $H\in \H^7(P)$ be an H-flux on $P$.  Suppose 
there exists a principal $\sfSU(2)$-bundle $\widehat P$ such that $\widehat c_2 \equiv c_2(\widehat P) = \pi_*H$.  Then
\begin{itemize}
\item[(i)] (Existence) there exists an $\widehat H\in \H^7(\widehat P)$ such that
\begin{equation} \label{eqn:eqBa}
\widehat \pi_* \widehat H = c_2 \,,\quad \text{and}\quad \widehat p^*H - p^* \widehat H = 0 \,,
\end{equation}
\item[(ii)] (Uniqueness) $\widehat H$ is uniquely determined by \eqref{eqn:eqBa} up to the addition of a term
$\widehat \pi^*( a \cup c_2)$, with $a\in \H^3(M)$.
\end{itemize}
\end{theorem}

\begin{proof} The correspondence \eqref{eqn:eqBAb} leads to the following commutative  square on cohomology
\begin{equation*}
\xymatrix{
& \H^7(P\times_M\widehat P) \\
\H^7(P) \ar[ur]^{\widehat p^*} && \H^7(\widehat P) \ar[ul]_{p^*}  \\
& \H^7(M) \ar[ul]^{\pi^*} \ar[ur]_{\widehat\pi^*}& 
}
\end{equation*}
We can complete this to the  double complex below.
\begin{equation*}
\xymatrix{
0 \ar[r]^{\cup \widehat c_2} \ar[d]^{\cup c_2} & \H^3(M) \ar[r]^{\widehat\pi^*}  \ar[d]^{\cup c_2}  & \H^3(\widehat P) \ar[r]^{\widehat\pi_*}  \ar[d]^{\cup \widehat\pi^* c_2} 
& \H^0(M) \ar[r]^{\cup \widehat c_2}  \ar[d]^{\cup c_2} & \H^4(M) \ar[r] \ar[d]^{\cup c_2} & \cdots \\
\H^3(M) \ar[r]^{\cup \widehat c_2} \ar[d]^{\pi^*} & \H^7(M) \ar[r]^{\widehat\pi^*} \ar[d]^{\pi^*}  & \H^7(\widehat P)\ar[r]^{\widehat\pi_*} \ar[d]^{p^*}   
& \H^4(M) \ar[r]^{\cup \widehat c_2} \ar[d]^{\pi^*}  & \H^8(M)\ar[r] \ar[d]^{\pi^*}  &  \cdots  \\
\H^3(P) \ar[r]^{\cup \pi^*\widehat c_2} \ar[d]^{\pi_*}  & \H^7(P) \ar[r]^{\widehat p^*} \ar[d]^{\pi_*}  & \H^7(P\times_M \widehat P) \ar[r]^{\widehat p_*} \ar[d]^{p_*} 
& \H^4(P)\ar[r]^{\cup \pi^* \widehat c_2} \ar[d]^{\pi_*}  & \H^8(P) \ar[r] \ar[d]^{\pi_*} & \cdots \\
\H^0(M) \ar[r]^{\cup \widehat c_2} \ar[d]& \H^4(M)\ar[r]^{\widehat\pi^*} \ar[d] & \H^4(\widehat P)\ar[r]^{\widehat\pi_*} \ar[d] & \H^1(M)\ar[r]^{\cup \widehat c_2} \ar[d] 
& \H^5(M) \ar[r] \ar[d]& \cdots \\
\vdots & \vdots &\vdots &\vdots &\vdots &
}
\end{equation*}
Without loss of generality we can assume that $M$ is connected, i.e.\ $\H^0(M)\cong \ZZ$.
\begin{itemize}
\item[(i)] 
Let $\widehat c_2  = \pi_* H$.  Then, since $p_* \widehat p^* H = \widehat\pi^* \pi_* H = \widehat\pi^* \widehat c_2 = 0$ we must
have $\widehat p^* H = p^* \widehat H$ for some $\widehat H\in \H^7(\widehat P)$.  Now, since 
$\pi^*  \widehat \pi_*  \widehat H = \widehat p_* p^* \widehat H = \widehat p_* \widehat p^* H = 0$ we have 
$ \widehat \pi_*  \widehat H = n\, c_2$ for some $n\in \H^0(M)$.  
But we can run this argument the other way around, starting with a $\widehat H$ satisfying $\widehat \pi_* \widehat H = c_2$, and
find an $H$ satisfying $\widehat p^* H = p^* \widehat H$.  For this $H$ we find $\pi_* H = m\, \widehat c_2$ for some $m\in \H^0(M)$. 
These two are only consistent for $m=n=1$.
\item[(ii)] Suppose $\widehat H$ and $\widehat H'$ both satisfy the requirements.  Let $h  = \widehat H' - 
\widehat H$.  We know $p^* h=0\in \H^7(E\times_M\widehat E)$, so $h=b\cup \widehat \pi^*c_2$, for some $b\in \H^3(\widehat P)$.
Consider $n=\widehat \pi_*b\in \H^0(M)$.  Now $(\widehat \pi_* a) \cup c_2 = n\, c_2  = \widehat \pi_* h = 0$, so either $c_2=0$, or $n=0$.
In the first case $h=0$ and $\widehat H$ is unique, or if $n=0$, then $b=\widehat \pi^* a$ with $a\in \H^3(M)$.  In that case
$h = \widehat \pi^* (a \cup c_2)$.  Clearly, any $h=\widehat \pi^* (a \cup c_2)$, $a \in \H^3(M)$, can be added to $\widehat H$ since $\pi_*(h) =0$.
\end{itemize}
\end{proof}

We claim that some of the distinct values of $\widehat H$ described in Theorem \ref{thm:thBAb} (ii) are related by automorphisms of $\widehat P$,
where $a\in \H^3(M)$ corresponds to the class of an automorphism of $\widehat P$
given by $g: M \to \sfSU(2)$.
Any automorphism $U: \widehat E \to \widehat E$ is given by some $g: M\to \sfSU(2)$ through the 
composition
\begin{equation*} \xymatrix{
U:\ \widehat P \ar[r]^{(\widehat \pi,1)} & M \times \widehat P \ar[r]^{g\times 1} & 
\sfSU(2) \times \widehat P \ar[r]^m & \widehat P  \,.}
\end{equation*}
We also have a pull-back diagram
\begin{equation*}
\xymatrix{
\sfSU(2) \times \widehat P \ar[r]^m  \ar[d]_{\text{pr}_2}  & \widehat P \ar[d]^{\widehat \pi} \\
\widehat P \ar[r]^{\widehat \pi}  & M \,,}
\end{equation*}
and by using the K\"unneth theorem,
\begin{equation*}
\H^7(\sfSU(2)\times \widehat P) \cong \text{pr}_2^* \H^7(\widehat P) \oplus (\Omega \times  \text{pr}_2^*  \H^4(\widehat P) ) \,,
\end{equation*}
where $\Omega$ is the canonical generator of $\H^3(\sfSU(2))$.
This implies
\begin{equation*}
m^*(\widehat H) =  \text{pr}_2^* \widehat H \oplus ( \Omega\times  \text{pr}_2^* \widehat \pi_* \widehat H) =
 \text{pr}_2^* \widehat H \oplus ( \Omega \times  \text{pr}_2^* \, c_2 ) \,,
\end{equation*}
where we have used $\widehat \pi_* \widehat H = c_2$. 
Now note that we have a map $[M, \sfSU(2) ] \to \H^3(M;\ZZ)$ through $[g] \mapsto g^*(\Omega)$.  In other words, if
our $a\in \H^3(M)$ is in the image of this map, then there exists an automorphism of $\widehat P$, given by $[g]\in [M,\sfSU(2) ]$
such that $[g] \mapsto a$ such that 
\begin{equation*}
U^*\widehat H = \widehat H + \widehat\pi^* (a \cup c_2)  = \widehat H'\,,
\end{equation*}
as asserted.

In general, however, contrary to the analogous case for $\sfU(1)$-principal bundles
where $[M, \sfU(1)] \cong \H^2(M;\ZZ)$,  the map $[M, \sfSU(2) ] \to \H^3(M;\ZZ)$ is neither injective, nor surjective.  
In the latter case, bundle isomorphisms do not exist which relate all of the allowed cohomology classes for $\widehat{H}$.  An example of non-injectivity
is given by $M=S^4$, where $[S^4, \sfSU(2)] \cong \pi_4(S^3) \cong \ZZ_2$, while $\H^3(S^4;\ZZ)\cong 0$.  An example
of non-surjectivity is given by $M=\sfSU(3)$.  Maps $\sfSU(3)\rightarrow\sfSU(2)$ induce homomorphisms $\H^3(\sfSU(3))\rightarrow\H^3(\sfSU(2))$ between the third cohomology groups, which are each isomorphic to the integers.  However the images of these maps only contain even elements (in physics this observation leads to a no go theorem for coloured dyons \cite{NM83}).

\subsection{Chern-Simons form}

In later sections of the paper we will see that the Chern-Simons  form \cite{CS74} plays a crucial role
in many of our considerations.  Here we give a brief overview of some of the results needed.

To each principal $\sfG$-bundle $\pi:P\to M$ is associated a 2nd Chern class $c_2(P)\in \H^4(M,\ZZ)$.
We have seen that in the case of $\sfG=\sfSU(2)$ this class enters crucially in the Gysin sequence of $P$,
which relates the cohomology of $P$ to the cohomology of the base space $M$.  

Now, let $A\in\Omega^1(P,\fg)$ be a principal connection on the principal $\sfG$-bundle $\pi:P\to M$, i.e.\ a connection which
reduces to the Maurer-Cartan  form $\Theta_{\text{MC}}$ of $\sfG$ on each fiber of $P$.  A de Rham representative of $c_2(P)$ is given by
\begin{equation}
c_2(P) = \frac{1}{8\pi^2} \text{Tr} ( F_A\wedge F_A) \,,
\end{equation}
where 
\begin{equation*}
F_A = dA + A\wedge A = dA + \frac12 [A,A] 
\end{equation*}
is the curvature of $A$.  Since $F \to g^{-1} F g$ under gauge transformations $g : M\to \sfG$, it follows
that $c_2(P)$ is actually a closed 4-form on $M$.  However, since the pull-back bundle $\pi^* P = P\times_M P$ under $\pi:P\to M$ is
trivial, the form $\pi^* c_2$ is exact on $P$ (this also follows from the Gysin sequence in the case $\sfG=\sfSU(2)$), hence we can
write
\begin{equation}
\pi^* c_2 = d\, \CS(A)\,,
\end{equation}
where $\CS(A)\in \Omega^3(P)$ is the so-called Chern-Simons 3-form.  An explicit expression for $\CS(A)$ is
\begin{equation}
\CS(A) = \frac{1}{8\pi^2}\   \text{Tr} ( A\wedge F_A - \frac13 A\wedge A\wedge A)  = \frac{1}{8\pi^2} \  \text{Tr} ( A\wedge dA + \frac23 A\wedge A\wedge A)\,.
\end{equation}
The Chern-Simons form is not gauge invariant, instead, under gauge transformations $g : M\to \sfG$ such that 
\begin{equation}
{}^gA = g^{-1} A g + g^{-1}dg \,,
\end{equation}
we have 
\begin{equation} 
\CS({}^g\!A) = \CS(A) + \CS(g^{-1} dg) +  \frac{1}{8\pi^2} \ d\, \text{Tr}  (g^{-1} dg\wedge g^{-1} A g) \,. \label{gaugex}
\end{equation}
where 
\begin{align}
\int_{S^3} \CS(g^{-1} dg) = -\frac{1}{24\pi^2} \int_{M} \text{Tr}  (g^{-1} dg \, \wedge g^{-1} dg\,  \wedge g^{-1} dg) = \text{deg}\, g
\end{align}
is the winding number (or degree)  of the map $g: M \to \sfG$.   In particular, for $g : \sfSU(2) \to \sfSU(2)$ given by the identity map, $g^{-1}dg$ can
be identified with the Maurer-Cartan form $\Theta_{\text{CM}}$ on $\sfSU(2)$, and thus principal connections on
principal $\sfSU(2)$-bundles are normalized precisely such that
\begin{equation}
\pi_* \CS(A) \equiv \int_{\sfSU(2)} \CS(A)  = 1 \,.
\end{equation}

\subsection{Cartan-Weil theory} \label{cwsez}

The Gysin sequence \eqref{eqn:eqBAa} indicates that the cohomology $\H^k(P)$ of a principal $\sfSU(2)$-bundle $\pi:P\to M$
is in some sense constructed out of the cohomologies $\H^k(M)$ and $\H^{k-3}(M)$ of the base $M$.  It will be useful
to have a similar statement at the level of forms.  The construction  we are about to review is called Cartan-Weil theory.

First consider a trivial principal $\sfG$-bundle $P = M \times \sfG$, where $\sfG$ is a 
connected compact Lie group.  Because of the K\"unneth
theorem we have 
\begin{equation*}
\H^k(P) \cong \bigoplus_{i+j=k} \left( \H^i(M) \otimes \H^j(\sfG) \right) \,.
\end{equation*}
The cohomology of $\sfG$ is generated by the so-called primitive elements $P_\sfG$ in $\H^\bullet(\sfG)$,
in the sense that $\H^\bullet(\sfG) \cong \bigwedge P_\sfG$.  
Primitive elements are those classes for which $\mu^*(\nu) = \nu\otimes 1 + 1 \otimes \nu$, where
$\mu^*$ is the pull-back of forms under multiplication $\mu: \sfG \times \sfG \to \sfG$, i.e.\
$\mu^* : \H^\bullet(\sfG) \to \H^\bullet(\sfG\times \sfG) \cong \H^\bullet(\sfG) \otimes \H^\bullet(\sfG)$.
For semi-simple compact groups $\sfG$, there are $n=\text{rank}\ \sfG$ primitive elements $\nu_i$ of
degree $2e_i+1$, where $\{ e_i\}_{i=1,\ldots,n}$ are the set of exponents of $\sfG$, and thus
$\text{dim}\ \H^\bullet(\sfG) = 2^n$.  Specific representatives of $\H^\bullet (P)$ can thus be chosen
to be of the form $\omega_i \otimes \nu_j$, where $\omega_i$, and $\nu_j$, are representatives of $\H^\bullet(M)$,
and $\H^\bullet(\sfG)$, respectively.  
Or, interpreted differently, the cohomology $\H^\bullet (P) = \H^\bullet(M\times\sfG)$ 
can be computed from the complex $(\Omega^\bullet(M)\otimes \bigwedge P_\sfG, D)$,
where $D=d_M$ is the de Rham differential on $\Omega^\bullet(M)$.  \medskip

In the case of a non-trivial principal $\sfG$-bundle the K\"unneth theorem no longer holds, and 
in general $\H^\bullet (P) \not\cong \H^\bullet(M) \otimes \H^\bullet(\sfG)$. Yet it turns out that the cohomology 
$\H^\bullet (P)$ can still be computed from a complex $(\Omega^\bullet(M)\otimes \bigwedge P_\sfG, D)$, albeit with a 
modified differential $D$.  More precisely, we have the following well-known theorem (see, e.g., \cite{GHV})
\begin{theorem} \label{thm:thBb}
There exists a quasi-isomorphism of graded differential complexes
\begin{equation*}
\Phi: \ (\Omega^\bullet(M)\otimes \bigwedge P_\sfG, D) \to (\Omega^\bullet(P), d)
\end{equation*}
known as the Chevalley homomorphism.
\end{theorem}

To describe $D$, we recall the following homomorphisms.  First of all there is a linear `transgression' map
\begin{equation*}
\tau: \ P_\sfG \to (S\fg^*)_{\text{inv}} \,,
\end{equation*}
which maps the primitive elements in $\H^\bullet(\sfG)$ to invariants in the symmetric algebra of $\fg^*$,
such that the primitive element of degree $2e_i+1$ maps to an invariant of degree $e_i+1$
(the so-called `Casimirs').  Next we have the Chern-Weil homomorphism
\begin{equation*}
h: \ (S\fg^*)_{\text{inv}} \ \to \ \H^\bullet(M) \,,
\end{equation*}
which doubles the degree, i.e.\ a Casimir invariant of degree $e_i+1$ maps to a $2(e_i+1)$-form, 
and involves the choice of a (principal) connection $A$ on $P$.
We choose a linear map $\gamma:\ (S\fg^*)_{\text{inv}} \ \to \ \Omega^\bullet(M)$ such that $\gamma(x)$
is a closed form representing the class $h(x)$, $x\in  (S\fg^*)_{\text{inv}}$.  The differential $D$ on 
$\Omega^\bullet(M)\otimes \bigwedge P_\sfG$ is then given by 
\begin{multline}
D (\omega \otimes (\nu_1\wedge \ldots \wedge \nu_p) )= 
(d_M\omega)\otimes (\nu_1\wedge \ldots \wedge \nu_p) \\  + (-1)^{\text{deg}\, \omega} \sum_{i=1}^p 
(-1)^{i-1} (\omega\wedge \gamma(\tau \nu_i) ) \otimes (\nu_1\wedge \ldots \widehat{\nu_i}  \ldots  \wedge \nu_p) \,.
\end{multline}
Finally, the Chevalley homomorphism $\Phi: \ \Omega^\bullet(M)\otimes \bigwedge P_\sfG \to \Omega^\bullet(P)$ is
given on $\Omega^\bullet(M)$ by the pull-back of $\pi: P\to M$, and on $P_\sfG$ by $\nu \to \widetilde\nu = \Phi(\nu)$,
where $\widetilde\nu\in\Omega^\bullet(P)$ is chosen such that $\pi^*\gamma(\tau \nu) = d\widetilde \nu$.

Rather than describing all these maps in detail, let us just make this more explicit in the case of $\sfG = \sfSU(2)$.
In this case there exists only one primitive element $\nu\in \H^3(\sfSU(2)) \cong \RR$, which can be taken as 
the volume form on $\sfSU(2) \cong S^3$.  To compute the cohomology of a principal $\sfSU(2)$-bundle $\pi: P\to M$,
we note that $\gamma( \tau \nu) = c_2(P) \in \H^4(M)$.  This involves the choice of a connection
$A$ on $P$.  The transgression map is the statement that the primitive form $\nu$, under the 
Chevalley homomorphism $\Omega^\bullet(M)\otimes \bigwedge P_\sfG \to \Omega^\bullet(P)$ is represented 
by the Chern-Simons form $\CS(A)$, where $\pi^* c_2(P) = d \, \CS(A)$.  Hence, to compute the cohomology $\H^k(P)$ of $P$,
it suffices to consider $k$-forms on $P$ of the type 
\begin{equation*}
\omega = \pi^* \omega_k +  \pi^* \omega_{k-3}  \wedge  \CS(A) \,,
\end{equation*}
with $\omega_i \in \Omega^i(M)$, with differential
\begin{equation*}
d\omega = \pi^* (d\omega_k + (-1)^k \omega_{k-3} \wedge c_2(P) )  +  \pi^* d\omega_{k-3} \wedge \CS(A)\,.
\end{equation*}
In Sec.~\ref{isosez} we will generalize this statement to twisted cohomology.

Note that in the case of a torus $\sfG = \TT^n$, the statement of Theorem \ref{thm:thBb} reduces to the 
treatment in \cite{BHM05}, but while in the 
case of the torus the complex agrees with $(\Omega^\bullet(P))_{\text{inv}} \cong \Omega^\bullet(M) \otimes \bigwedge \fg^*$
for $\sfG=\sfSU(2)$ the space of invariant forms on $P$ is a lot bigger than $\Omega^\bullet(M)\otimes \bigwedge P_\sfG$.  
[The map $\fg^* \to (\Omega^1(P))_{\text{inv}}$ is through a principal connection $A\in\Omega^1(P,\fg)$.]

For those not familiar with Cartan-Weil theory, it is illuminating to apply 
the above computation to the principal $\sfSU(n)$-bundles $\sfSU(n+1) \to S^{2n+1}$. By induction one finds
\begin{equation*}
\H^\bullet( \sfSU(n+1) ) \cong \H^\bullet(S^3 \times \ldots \times S^{2n+1}) \,.
\end{equation*}

\medskip

\section{Construction of the spherical T-Dual: Classifying Space Approach}

We will now present an alternate approach to the construction of the T-dual.  In Subsec.~\ref{sec:sphericalTduality} we present a classifying space $R$ for a pair $(P,H)$  consisting of a principal $\sfSU(2)$-bundle $P\rightarrow M$ and $H\in\H^7(P, \ZZ)$.   This construction is similar in spirit 
to that  for $\sfU(1)$ bundles which was introduced in \cite{RR88} and studied in detail by Bunke and Schick \cite{BS}. 
The classifying space for classical T-duality of higher rank principal torus bundles was later studied in \cite{MR, MR2, MR3,BS2}.  
However in the present case,  in general a map $M\rightarrow R$ no longer uniquely defines the T-dual bundle 
$\widehat P$.   These instead are represented by the classifying space $S$, introduced in Subsec.~\ref{ssez}.  A pair $(P,H)$ on $M$ will be T-dualizable only if a map $M\rightarrow R$ lifts to a map $M\rightarrow S$ and there will be a distinct T-dual for each distinct lift.  On the other hand, in Subsec. \ref{ratsez} we will see that rationally $R$ and $S$ are homotopy equivalent and so, rationally, T-duals exist and are unique.  In Subsec. \ref{interpsez} we provide an interpretation for the rational homotopy theory approach.  Finally in Subsec.~\ref{qsez} we show that if $M$ is a 4-manifold then T-duals always exist and are unique.

\subsection{Classifying space of pairs}
\label{sec:sphericalTduality}

Here we will assume that $\sfG = \sfSU(2)$ (although some of what we show holds for more general simply-connected compact Lie groups).
Let $M$ be a locally compact Hausdorff 
space with the homotopy type of a finite CW complex.  
A \emph{pair} $(P, H)$ over $M$ will mean a principal
$\sfG$-bundle $P \to M$ together with a class $H \in \H^7(P,\ZZ)$ (note that for
dimension reasons, the restriction of $H$  to each $\sfG$ fiber is $0$).

\begin{theorem}
\label{thm:classspace}
The set of pairs $(P, H)$ as defined above, modulo isomorphism, is a 
representable functor, with representing space
\begin{equation}
\label{eq:ClassSpace} 
R = E\sfG \times_\sfG \Maps(\sfG, K(\ZZ, 7)), 
\end{equation}
where $E\sfG \to B\sfG$  is the universal $\sfG$-bundle. 
Note that  $R$ is path-connected since $H^7(\sfG;\ZZ)=0$.  
Then $R$ is a fiber bundle
\beq
\Maps(\sfG, K(\ZZ, 7))\longrightarrow R \stackrel{c}{\longrightarrow}  B\sfG
\eeq
with fiber $\Maps(\sfG, K(\ZZ, 7))$.
\end{theorem}

\begin{proof}
Construct
a tautological pair $(\bE, \bh)$ over $R$ by setting 
$\bE = E\sfG \times \Maps(\sfG, K(\ZZ, 7))$ 
with the diagonal action of $\sfG$.
Then define $\bh\co \bE
\to K(\ZZ, 7)$ by the formula
\beq
\bh(u,\, [v, \gamma]) = \gamma(g),
\eeq
where $\gamma\in \Maps(\sfG, K(\ZZ, 7))$, $g\in \sfG,$ $u,\,v\in EG$,
$u$ and $v$ live over $c([v, \gamma])$, and $gv=u$. One can check 
that this is independent of the choices of $u$, $v$, and $\gamma$
representing a particular element of $\bE$. Clearly any map $f:M\to R$
enables one to pull back the canonical pair $(\bE, \bh)$ to a pair over
$M$. 

On the other hand, suppose we have a pair $(P,H)$ over
$M$. Since $P\xrightarrow{\pi} M$ is a principal $\sfG$-bundle, we know
that $P\xrightarrow{\pi}  M$ is 
pulled back from the universal bundle $E\sfG\to B\sfG$ via a map $\phi\co M \to
B\sfG$. We claim we can fill in the diagram
\[
\xymatrix{ &K(\ZZ, 7)& \bE\ar[d] \ar[l]_(.35)\bh \ar[dl]\\
P \ar[d]^\pi \ar[ru]^h \ar[r]_{\wp}  \ar@{.>}[rru]^{\wf}
& E\sfG \ar[d] & R \ar[ld]^(.45)c\\
M \ar[r]_(.45)\varphi \ar@{.>}[rru]^(.4){f}  |\hole& \,B\sfG . &}
\]
as shown, to make it commute, and to realize $(P,H)$ as the pull-back
of $(\bE, \bh)$. Indeed, we simply define $f(z)=[\wp(e), \gamma]$,
where $e\in \pi^{-1}(z) \subseteq P$, and where $\gamma\in 
\Maps(\sfG, K(\ZZ, 7))$ is defined by $\gamma(g)=h(g\cdot e)$. 
Note that $f(z)$ is independent of the choice of
$e$. We can define $\wf$ by $\wf(e)=\bigl[\wp(e),[\wp(e), \gamma] 
\bigr]$, with $e$ as before. The rest of the proof is as in \cite{BS}.
\end{proof}

Next we want to understand the homotopy type of $R$.
We start by studying the homotopy type of the fiber $\Maps(\sfG, K(\ZZ, 7))$. 
Note that 
\beq
\Maps(\sfG, K(\ZZ, 7)) \sim \Maps^+(\sfG, K(\ZZ, 7)) \times K(\ZZ, 7) \sim K(\ZZ, 4) \times K(\ZZ, 7).
\eeq
Therefore $R$ can be realized as a homotopy fibration 
\beq
K(\ZZ, 4) \times K(\ZZ, 7) \stackrel{c}{\to} R \to B\sfG,
\eeq
or equivalently, as will be used below, as a homotopy fibration 
\beq
K(\ZZ, 7) \to R \stackrel{\mp}{\to} K(\ZZ, 4) \times B\sfG.
\eeq

\subsection{Classifying space of spherical T-dual pairs of pairs} \label{ssez}

In the last subsection we saw that maps $f:M\rightarrow R$ classify pairs $(P,H)$ where 
$\pi:P\rightarrow M$ is a principal $\sfSU(2)$-bundle on $M$ and $H\in\H^7(P, \ZZ)$.  
However not all pairs have spherical T-duals and when they do, the dual is not necessarily unique.   
In this subsection we describe another classifying space, $S$, which classifies spherical 
T-dual pairs of pairs. In the case of the circle bundles, these classifying spaces coincide and 
topological T-duality works perfectly. However in our case, $S$ and $R$ are not homotopy equivalent.  Instead, there is only a map $g:S\rightarrow R$.  While pairs $(P,H)$ correspond to maps $f:M\rightarrow R$, in this subsection we will show that T-duals correspond to lifts $\tilde f:M\rightarrow S$ such that $g\tilde f=f$.

First, recall that $R$ is a $K(\Z,7)$ bundle over $K(\Z,4)\times B\sfG$.  Let $\widehat\alpha$ and $\alpha$ 
generate $\H^4(K(\Z,4), \Z)\cong\Z$ and $\H^4(B\sfG, \Z)\cong\Z$ respectively.  Then the characteristic class 
of the fibration $\mp:R\to K(\Z,4)$ is the $k$-invariant $\alpha\cup\widehat\alpha$.  Consider the map $g:B\sfG_1\times B\sfG_2\rightarrow 
K(\Z,4)\times B\sfG$ which corresponds to the generator of 
$\H^4(B\sfG_1, \Z) = [B\sfG_1, K(\Z, 4)]\cong\Z$ on $B\sfG_1$ and which is the identity on $B\sfG_2$.  
Here $B\sfG_1$ and $B\sfG_2$ are two copies of $B\sfG$.

\begin{theorem}
The set of spherical T-dual pairs $(P, H)$ and $(\widehat P, \widehat H)$, modulo isomorphism, is a representable functor, with 
representing space $S$, defined to be $P:S\rightarrow B\sfG_1\times B\sfG_2$,  the  homotopy $K(\Z, 7)$-fibration 
over $B\sfG_1\times B\sfG_2$ pulled back from $\mp:R\rightarrow K(\Z,4)\times B\sfG$ via the map $g$, that is, $S=g^*R$.
\end{theorem}
\begin{equation*}
\xymatrix{
K(\Z,7)\ar[r] &S\ar[d]_P&K(\Z,7)\ar[r] &R\ar[d]_\mp\\
& B\sfG_1\times B\sfG_2\ar[rr]^g&& K(\Z,4)\times B\sfG
}
\end{equation*}
The $k$-invariant of the $K(\Z,7)$ fibration $S$ is $\beta\cup\widehat\beta \in \H^8(B\sfG_1\times B\sfG_2;\ZZ)$ where $\beta=g^*\alpha$ and $\widehat\beta=g^*\widehat\alpha$ generate $\H^4$ of the two copies of $B\sfG$.  Lifting $g$ to the total spaces of the $K(\Z,7)$ fibrations one obtains an induced map $g : S \to R$ which we denote by the same symbol.

\begin{proof}
We can define two principal $G$-bundles over $S$, $\Pi:\bF\rightarrow S$ and $\widehat \Pi:\widehat \bF\rightarrow S$ 
to be the pullback of the universal $G$ bundle $EG\rightarrow BG$ to $S$ via the projection map $P$ composed 
with the projection maps $B\sfG_1\times B\sfG_2\rightarrow B\sfG_1$ and $B\sfG_1\times B\sfG_2\rightarrow B\sfG_2$ respectively.   
Note that $c_2(\bF)=P^*\beta$ and $c_2(\widehat \bF)=P^*\widehat\beta$.  
\begin{equation*}
\xymatrix{
P\ar[d]_{\pi}\ar[r]^{\tilde f}&\bF\ar[d]_{\Pi}\ar[r]^g&\bE\ar[d]\\
M \ar[r]^{\tilde f}& S\ar[r]^g&R
}\hspace{.9cm}
\xymatrix{
\hat P\ar[d]_{\hat\pi}\ar[r]^{\tilde f}&\hat \bF\ar[d]_{\hat \Pi}\ar[r]^{i}& \bF\ar[d]_{\Pi}\ar[r]^{g}& \bE\ar[d]\\
M \ar[r]^{\tilde f}&S \ar[r]^{i}& S\ar[r]^g&R
}
\end{equation*}
Consider a map $\tilde f:M\rightarrow S$.   One can obtain two principal $\sfSU(2)$-bundles $\pi:P\rightarrow M$ and 
$\widehat\pi:\widehat P\rightarrow M$ by pulling back those over $S$, $P=\tilde f^*\bF$ and 
$\widehat P=\tilde f^*\widehat \bF$.  Notice that $g\tilde f:M\rightarrow R$ and so, by 
Theorem \ref{thm:classspace}, $g\tilde f$ yields a pair $(P,H)$ with $H\in\H^7(P; \Z)$.  
As $g$ is the identity on $BG_2$, $\bF=g^*\bE$ and so the $P=\tilde f^*\bF$ obtained from $f$ is identical to $P=(g\tilde f)^*\bE$ obtained from $gf$.   

Let $i : B\sfG_1 \times B\sfG_2 \to B\sfG_1 \times B\sfG_2$ be the (homotopy) involution which exchanges $B\sfG_1$ and $B\sfG_2$.
Since it preserves the $k$-invariant of $S$, it lifts to a (homotopy) involution $i:S\rightarrow S$. Then 
$gi\tilde f:M\rightarrow R$ yields a pair $(\widehat P,\widehat H)$.  Therefore a map $\tilde f:M\rightarrow S$ yields two pairs $(P,H)$ and $(\widehat P,\widehat H)$.  

To be spherical T-dual these pairs need to satisfy
\beq
c_2(P)=\widehat\pi_*\widehat H\hsp c_2(\widehat P)=\pi_* H\hsp p^*\widehat H=\widehat p^* H.
\eeq
By naturality, these relations are the pullbacks $\tilde f^*$ of the corresponding statements on $S$
\beq
c_2(\bF)=\widehat \Pi_*\widehat h\hsp c_2(\widehat \bF)=\Pi_* h\hsp (\Pi\otimes 1)^*\widehat h=(1\otimes\widehat \Pi)^* h
\eeq
where $h\in\H^7(\bF, \Z)$ and $\widehat h\in\H^7(\widehat \bF, \Z)$ are pulled back from $\bE$ 
via $h=g^*\bh$ and $\widehat h=(gi)^*\bh$.  The third property is a consequence of the fact that $i$ induces 
a homeomorphism on the correspondence space 
$\bF\times_S\widehat \bF$.  At low dimensions only the $S^7\subset K(\Z,7)$ appears in the skeleton of 
$\bF$ and so the Gysin sequence can be used to demonstrate the first two properties as follows.  The facts
\beq
P^*\beta\cup P^*\beta\neq 0\in\H^8(S,\Z)\hsp P^*\widehat\beta\cup P^*\widehat\beta\neq 0\in\H^8(S,\Z)\hsp
P^*\beta\cup P^*\widehat\beta= 0\in\H^8(S,\Z)
\eeq
imply that $\im(\Pi_*:\H^7(\bF,\Z)\rightarrow\H^4(S,\Z))=\ker(P^*\beta\cup)$ is generated by $P^*\widehat\beta=c_2(\widehat \bF)$ and so, up to a sign which can be fixed by an automorphism, $c_2(\widehat \bF)=\Pi_* h$.  The T-dual statement follows by an identical argument using the Gysin sequence for $\widehat P$.  So we have shown that the pairs $(P,H)$ and $(\widehat P,\widehat H)$ corresponding to the map $M\rightarrow S$ are indeed T-dual.

We will say that a pair $(P,H)$ is {\it T-dualizable} if the corresponding map $f:M\rightarrow R$ lifts to a map $\tilde f:M\rightarrow S$ such that $g\tilde f=f$.  
\begin{equation*}
\xymatrix{
& S\ar[d]_{g} \\
M \ar@{.>}[ur]^{\tilde f}\ar[r]^{f} & R
}
\end{equation*}

In the case of T-duality of circle bundles, the homotopy equivalence between $K(\ZZ,2)$ and $B\sfU(1)$ induced an 
equivalence between the analogues of $R$ and $S$, and so pairs of 
$\sfU(1)$ bundles with 3-cocycles in the total space are always T-dualizable.  In the present case $B\sfSU(2)$ is not a model for $K(\Z,4)$ as $\sfSU(2)$ is not a model for $K(\Z,3)$ and so the lift may fail to exist or to be unique.  However there is one spherical T-dual for each lift $\tilde f$.

Are all spherical T-dual pairs of pairs representable by maps to $S$?  Consider a T-dual pair of pairs 
$(P,H)$ and $(\widehat P,\widehat H)$.  By Theorem \ref{thm:classspace} corresponding to $(P,H)$ there is a map 
$f:M\rightarrow R$ such that $P=f^*\bE$ and $H=f^*\bh$.   As $(\widehat P,\widehat H)$ is T-dual to $(P,H)$
\beq
c_2(\widehat P)=\pi_*(H)=\pi_*f^*\bh=f^*\pi_*\bh=f^*\mp^*\widehat\alpha.
\eeq
The bundles $\widehat\pi:\widehat P\rightarrow M$ and $\pi:P\rightarrow M$ can be pulled back from the universal 
$\sfG$-bundle and so are represented by a map $\phi:M\rightarrow B\sfG\times B\sfG$ such that 
$(\widehat P,P)=\phi^* (E\sfG\times B\sfG,B\sfG\times E\sfG)$.  In particular, as $\widehat\beta$ and $\beta$ are the generators the two copies of 
$\H^4(B\sfG)$, the Chern classes can be expressed in terms of their pullbacks as $\phi^*(\widehat\beta,\beta)=(c_2(\widehat P),c_2(P))$.  Recall that $g:B\sfG_1\times B\sfG_2\rightarrow K(\Z,4)\times B\sfG $ 
is the map representing the generator of $\H^4(B\sfG_1)$ on $B\sfG_1$ and the identity on $B\sfG_2$.  Then
\beq
\phi^*g^*(\widehat\alpha,\alpha)=\phi^*(\widehat\beta,\beta)=(c_2(\widehat P),c_2(P)) = (f^*\mp^*(\widehat\alpha),f^*\mp^*(\alpha)).
\eeq
Thus we learn that $\phi:M\rightarrow B\sfG_1\times B\sfG_2$ lifts $\mp f:M\rightarrow B\sfG\times K(\Z,4)$ so that $\mp f=g\phi$.  
\begin{equation*}
\xymatrix{
& B\sfG_1\times B\sfG_2 \ar[dr]_{g}&R\ar[d]^\mp \\
M \ar[ur]^\phi\ar[urr]_{f}\ar[rr]_{\mp\circ f=g\circ\phi} & & B\sfG\times K(\Z,4)
}
\end{equation*}

To prove that this pair of pairs is representable, we need to construct not $\phi$ but rather its lift $\tilde f:M\rightarrow S$, where $\phi=P\tilde f$, $P:S\rightarrow B\sfG_1\times B\sfG_2$. 
\begin{equation*}
\xymatrix{
&S\ar[d]_P\ar[dr]_g&\\
& B\sfG_1\times B\sfG_2 \ar[dr]_{g}&R\ar[d]^\mp \\
M  \ar@{.>}[uur]^{\tilde f} \ar[ur]^\phi\ar[urr]_{f}\ar[rr]_{\mp\circ f=g\circ\phi} & & B\sfG\times K(\Z,4)
}
\end{equation*}
The obstruction to this lift is just the pullback of the $k$-invariant characteristic class of $P:S\rightarrow  B\sfG_1\times B\sfG_2$, 
\beq
\phi^*(\beta\cup\widehat\beta)=c_2\cup\widehat c_2=0
\eeq
where the last equality is a consequence of the fact that $(P,H)$ and $(\widehat P,\widehat H)$ are T-dual together with the Gysin sequence for either $P\rightarrow M$ or $\widehat P\rightarrow M$.  The obstruction is just the characteristic class of the restriction of the $K(\Z,7)$ bundle $S\rightarrow BG_1\times BG_2$ to $\phi(M)$; as it vanishes, this restricted bundle is trivial.  Therefore any section $\tilde f:M\rightarrow S$  of this trivial bundle provides a lift of $\phi$ and so exists, although in general it will not be unique.   As $\tilde f$ lifts $f$ and $f$ represents $(P,H)$, one can pull back the 7-class $h$ from $\bF$ to $P$ to obtain $H$ as desired.

Where is the ambiguity in choosing $\widehat H$?  Begin with a pair $(P,H)$, which determines a map $f:M\rightarrow R$.  Recall that there will be a T-dual pair $(\widehat P,\widehat H)$ for each lift $\tilde f:M\rightarrow S$ such that $g\tilde f=f$.  In particular $\widehat P=\tilde f^*\widehat \bF$.   We construct the dual twist as
\beq
\widehat H=\tilde f^* h=\tilde f^*i^*g^*\bh\hsp \tilde f:\widehat P\rightarrow\widehat \bF\hsp i:\widehat \bF\rightarrow \bF\hsp g:\bF\rightarrow \bE
\eeq
where we have lifted everything to the dual bundles.  

The dual twist is therefore determined by the map $\tilde f^*:\H^7(\widehat \bF)\rightarrow\H^7(\widehat P)$.  The Gysin sequence for $\widehat F$ yields
\beq
\xymatrix{
0\ar[r]^{\pi^*} & \H^7(\widehat \bF)\ar[r]^{\pi_*} & \H^{4}(S)\ar[r]^{\widehat\beta\cup} & \H^{8}(S). 
}\
\eeq  
As $\ker(P^*\widehat\beta\cup \cdot :\H^4(S)\rightarrow\H^8(S))=\beta\Z\cong\Z$ and $\pi_*:\H^7(\widehat \bF)\rightarrow\H^4(S)$ is an isomorphism onto this kernel, we find $\H^7(\bF)\cong\Z$.  It is generated by $\widehat h=i^*h$.  So $\widehat h$ is well defined, up to a sign.  However $\widehat H = \tilde f^*h$ is not unique determined as it depends upon the choice of lift $\tilde f$.  As $f=g\tilde f$ we know that $\tilde f^* g^*=f^*$ but there is no formula for $\widehat H$ as $f^*$ of a class on a bundle over $R$, so this does not uniquely determine $\widehat H$.  

However we know that on the correspondence space $\widehat P\times_M P$, $p^*\widehat H=\widehat p^* H$ which can be pulled back from a corresponding identity on the correspondence space $\widehat \bF\times_S \bF$
\beq
(\Pi\otimes 1)^*\widehat h=(1\otimes\widehat\Pi)^* h\in\H^7(\widehat \bF\times_S \bF).
\eeq
As a result, $p^*\widehat H$ is known, and so the difference $\tilde f_1^*(h)-\tilde f_2^*(h)$ in the values of $\widehat H$ defined via two distinct lifts must be in
\beq
\ker\left(p^*:\H^7(\widehat P)\rightarrow\H^7(\widehat P\times_M P)\right)\cong \widehat\pi^*c_2(P)\cup\H^3(\widehat P)\cong \widehat\pi^*(c_2(P)\cup\H^3(M))
\eeq
where in the last step we assumed $c_2(\widehat P)\neq 0$.   This matches the set of differences of admissible values of $\hat H$ found in Theorem \ref{thm:thBAb} (ii), and so for any choice of T-dual pair $(\widehat P,\widehat H)$ there always exists a lift $\tilde f:M\rightarrow S$ such that $\tilde f^*h$ will be equal to the desired value of $\hat H$.
\end{proof}

Do the different values of $\widehat H$ correspond to distinct T-dual pairs $(\widehat P,\widehat H)$?  In the circle bundle case such an ambiguity can be absorbed via an automorphism of the bundle \cite{BS}.  In the present case
\begin{lemma}
While a pair $(P,H)$ does not necessarily uniquely determine a spherical T-dual $\sfSU(2)$-bundle $\widehat P$, it does determine a T-dual 2-gerbe.  
\end{lemma}

To see this, consider the associated map $f:M\rightarrow R$.  The $K(\Z,4)$ in $R$ is the classifying space for 2-gerbes.  Thus, a 2-gerbe with characteristic class $P^*\widehat\alpha$ can be pulled back from $R$ to $M$ via $f$.  Consider an element $a\in\H^3(M, \Z)$.  This can be represented by a map $\phi:M\rightarrow K(\Z,3)$ or equivalently by a 1-gerbe on $M$ with characteristic class $a$.  Trivial 2-gerbes are classified by 1-gerbes and trivial 2-gerbes act as automorphisms on 2-gerbes that do not change their characteristic class.  Therefore, automorphisms $A_a$ of the 2-gerbe are isomorphic to the group $\H^3(M)$ of values of $a$.   

Realizing the 2-gerbe as a $K(\Z,3)$ bundle and the 1-gerbe as a map $\theta:M\rightarrow K(Z,3)$, this automorphism is just the fiberwise multiplication of the $K(\Z,3)$ over each point $m\in M$ by $\theta(m)$.  Realizing $\widehat H$ as a 7-cocycle of the total space of the 2-gerbe which pulls back from an $\sfSU(2)$-subbundle, if it exists, by the Gysin sequence $\widehat H$ can be split into a component which pulls back from $M$ and a component which pushes forward to yield $\widehat\pi_*\widehat H=c_2$.   The automorphism changes the choice of splitting because $\pi_*(A_a^*\pi^*)\neq 0$.  The choice of splitting changes by $a$, so $(\widehat\pi_*)^{-1}(1)$ is increased by $\widehat\pi^*a$.  Therefore $(\widehat\pi_*)^{-1}(c_2)$ is increased by $\widehat\pi^*(c_2\cup a)$, resulting in $\widehat H\mapsto\widehat H+\widehat\pi^*(c_2\cup a)$.  As a result the ambiguity in the choice of $\widehat H$ can be removed by a transformation of the 2-gerbe.

However we are not interested in the 2-gerbe itself, but the $\sfG$-bundle $\widehat P$ to which it lifts when $(P,H)$ is T-dualizable.  Let us restrict our attention to $\sfG=\sfSU(2)$.  Given an $\sfSU(2)$-bundle we can naturally associate a 2-gerbe \cite{danny}.  Thus we can pull the 2-gerbe over $R$ back to $S$ using $g$.  The characteristic class of the pulled back 2-gerbe will be equal to $c_2(\bF)\in\H^4(S)$ and so the T-dual  will be the 2-gerbe associated to $\widehat \bF\rightarrow S$.   Furthermore if $(P,H)$ is T-dualizable then a lift $\tilde f$ exists which we can use to pull the 2-gerbe back to $M$, where it will be the 2-gerbe associated to $\widehat P$.  The 1-gerbe automorphism with characteristic class $a$ can also be pulled back to $M$, indeed $a\in\H^3(M)$.    How can we act the 1-gerbe on $\sfSU(2)$ fibers?  Again consider a 2-gerbe to be a $K(\Z,3)$ bundle and a 1-gerbe a map to $K(\Z,3)$.  Now the 1-gerbes act on 2-gerbes by fiberwise multiplication.  

There is a rather nice explicit description not involving 2-gerbes in the case of interest to string theory. Let us begin by recalling the following well known fact.
\begin{lemma} If $dim(M) \le 13$, we can realize $K(\Z,3)$ by the Lie group $\sfE_8$ and the automorphism will simply correspond to the fiberwise group multiplication in $\sfE_8$.   In particular, principal $\sfE_8$ bundles over $M$ are 
classified up to isomorphism by their first Pontryagin class, $p_1=c_2 \in \H^4(M, \Z)$.
\end{lemma}

\newcommand{\longhookrightarrow}{\ensuremath{\lhook\joinrel\relbar\joinrel\rightarrow}}

We are interested in understanding more explicitly what it means for principal $\sfSU(2)$-bundles over $M$ to be 
equivalent as 2-gerbes, or equivalently, when they have the same 2nd Chern class $c_2$. 
Consider $\sfSU(2)$ that is minimally embedded into $\sfE_8$, $j\colon \sfSU(2)\longhookrightarrow\sfE_8$, for example by considering a single simple root, so that the embedding induces an isomorphism of the third cohomology groups. 
Then we deduce,

\begin{lemma}
Let $P_k, \, k=1, 2$ denote principal $\sfSU(2)$-bundles over $M$, where $dim(M)\le 13$. Define the associated principal 
$\sfE_8$-bundles $Q_k=P_k \times_{\sfSU(2)} \sfE_8$ using the faithful homomorphism $j$ as above. Then 
$P_k, \, k=1, 2$ are equivalent as 2-gerbes if and only if $Q_k, \, k=1, 2$ are isomorphic principal $\sfE_8$-bundles. In particular if and only if $c_2(P_1)=c_2(P_2)$.
\end{lemma}


Now can we use a map $\psi:M\rightarrow \sfE_8$ to create an automorphism on an $\sfSU(2)$-bundle?  The embedding of $\sfSU(2)$ in $\sfE_8$ gives a multiplication rule $\sfE_8\times\sfSU(2)\rightarrow \sfE_8$.  Thus the $\sfE_8$ action on an $\sfSU(2)$-bundle creates a new bundle with transition functions in $\sfE_8$, not an $\sfSU(2)$ subbundle as desired.  However we would get an $\sfSU(2)$ subbundle if we could lift $\psi$ to $\tilde\psi:M\rightarrow \sfSU(2)$.   The obstruction to removing the ambiguity in the construction of $\widehat{H}$ with an $\sfSU(2)$ bundle automorphism is just the obstruction to the existence of the lift $\tilde\psi$ satisfying $j\tilde\psi=\psi$.

\begin{equation*}
\xymatrix{
& \sfSU(2)\ar[d]_{j} \\
M \ar@{.>}[ur]^{\tilde \psi}\ar[r]^{\psi} & \sfE_8
}
\end{equation*}

More generally, in any dimension $a\in\H^3(M)$ can be represented by a map $F:M\rightarrow K(\Z,3)$.  The ambiguity in the definition of $\widehat H$ corresponding to $a$ can be eliminated by a bundle automorphism $\tilde F:M\rightarrow \sfSU(2)$ if and only if this 1-gerbe $a$, which generates the automorphism on the 2-gerbe, lifts to $\tilde F$ such that the pullback $\tilde F^*$ of the top class of $\sfSU(2)$ is equal to the characteristic class $a$ of the 1-gerbe.   If $J:\sfSU(2)\rightarrow K(\Z,3)$ represents the generator of $\H^3(\sfSU(2))$ then the ambiguity in $\widehat H$ corresponding to $a$ can be undone via an 
$\sfSU(2)$ bundle automorphism if there exists a lift $\tilde F$ such that $J \tilde F=F$
\begin{equation*}
\xymatrix{
& \sfSU(2)\ar[d]_{J} \\
M \ar@{.>}[ur]^{\tilde F}\ar[r]^{F} & K(\Z,3)
}
\end{equation*}
Rationally the existence of the automorphism is equivalent to the existence of change of connection (\ref{gaugex}) of the $\sfSU(2)$ bundle which shifts the Chern-Simons form on $\widehat P$ by any integral period closed form corresponding to $a$.  The obstructions to these objects are the obstruction to $\widehat H$ being defined up to bundle automorphism given a pair $(P,H)$.

\subsection{Rational homotopy theory approach} \label{ratsez}

In the previous 2 subsections, we constructed a pair of classifying spaces $(R, S)$, where $R$ consists of 
pairs $(P, H)$ over $M$ consisting of a principal
$\sfG$-bundle $P \to M$ together with a class $H \in \H^7(P,\ZZ)$,  $\sfG=\sfSU(2)$ as before, 
and $S$ consists of spherical T-dual pairs of such pairs. The problem with spherical T-duality
in higher dimensions can be encapsulated by the observation that $R\ne S$. Using rational homotopy theory,
we observe that 
the rationalizations $R_\QQ = S_\QQ$ are equal and so spherical T-duality works nicely over the rationals.
We will explain what this actually means in the next section. 

Recall that we have a homotopy fibration 
\beq
K(\ZZ, 7) \to R \stackrel{\mp}{\to} K(\ZZ, 4) \times B\sfG.
\eeq
Therefore 
\beq
\pi_j(R)=\pi_j(B\sfG), \qquad\text{if}\,\, j<4 \,\,\,\text{or}\,\,\, j>8  \,\,\,\text{or}\,\,\, j=6.
\eeq
In the other cases, one has that $\pi_j(R)$ is an extension of $\pi_j(B\sfG)$ by $\ZZ$.
Herein lies a serious problem, namely since the homotopy groups of spheres and 
in particular of $\sfG$ are unknown, therefore the homotopy groups of $B\sfG$ are also 
unknown. It follows that the homotopy groups
of $R$ are also unknown!

However there is a partial fix, given by Quillen's {\em rational homotopy theory} \cite{Q} of simply connected spaces,
motivated by the well known result of Serre that the homotopy groups of spheres when tensored over the rationals,
are completely understood. For more details on this theory, see also \cite{FHT}, \cite{KH}.

The {\em rationalization} of a simply connected space $X$, denoted $X_\QQ$, has homotopy groups 
$\pi_j(X_\QQ) = \pi_j(X) \otimes \QQ.$ For example, $\sfG_\QQ = K(\ZZ, 3)_\QQ$ and 
$B\sfG_\QQ = K(\ZZ,4)_\QQ$.

So instead we study the rational homotopy type of $R$, denoted $R_\QQ$. 

\begin{theorem}\label{thm:homtypeofR}
The space $R_\QQ$  has only two non-zero homotopy groups,  
$\pi_4(R_\QQ)\cong \QQ^{2}$, and $\pi_7(R_\QQ)\cong  \QQ$. 
\end{theorem}
\begin{proof}
This follows from the rational homotopy fibration,
\beq
K(\ZZ, 4)_\QQ \times K(\ZZ, 7)_\QQ \to R_\QQ \to K(\ZZ,4)_\QQ.
\eeq
\end{proof}

The space $R_\QQ$
has only two non-zero homotopy groups, $\pi_4$ and $\pi_7$, and so it
is a two-stage Postnikov system just like in \cite{BS}.

\begin{theorem}[Universal rational spherical T-duality]
\label{thm:geomTduality}
The rationalization $R_\QQ$ of the classifying space $R$ 
is a two-stage Postnikov system 
\beq\label{exact2}
K(\ZZ,7)_\QQ \to R_\QQ \to K(\ZZ,4)_\QQ \times K(\ZZ,4)_\QQ,
\eeq
with $\pi_7(R_\QQ)\cong \QQ$ and
with $\pi_4(R_\QQ)\cong \QQ\oplus \QQ$. 

The $k$-invariant of
$R_\QQ$ in $\H^8(K(\bZ,4)_\QQ\times K(\bZ,4)_\QQ,\QQ)$ can be identified with
$x\cup y$, where $x$ is any nonzero element of $\H^4$ of the first copy of $K(\bZ,4)_\QQ$ and
$y$ is the same element of $\H^4$ of the second copy of $K(\bZ,4)_\QQ$. Rationally, spherical T-duality is 
implemented by a self-map $\#$ of $R_\QQ$, whose square is homotopic to the
identity, interchanging the two copies of $K(\bZ,4)_\QQ$. {\lp}The involutive 
automorphism of $K(\bZ,4)_\QQ\times K(\bZ,4)_\QQ$ interchanging the two
factors preserves the $k$-invariant and thus extends to a homotopy
involution of $R_\QQ$.{\rp}
\end{theorem} 

\begin{proof}
We have already computed the homotopy groups of $R_\QQ$. 

To finish the proof of the theorem, we need to check that the
$k$-invariant of $R_\QQ$ is as described. The proof of this fact
is similar to \cite{BS}, \cite{MR2} once we have the fibration \eqref{exact2}.
\end{proof}

\begin{lemma}
\label{lem:cohomcalc}
Let $R$ be the classifying space for pairs $(P, H)$, as in Theorem \ref{thm:geomTduality}, and let
$(p\co \bE\to R,\bh)$ be the canonical pair over $R$. Then upon rationalization, we get pairs 
$(\bE_\QQ\to R_\QQ, \bh_\QQ)$.
Then with notation as in Theorem \ref{thm:geomTduality}, 
the cohomology ring 
\[
\H^*(R_\QQ) =\QQ[x,y]/(xy) 
\]
where $x$ and $y$ are in degree $4$ and in particular,
$\H^j(R_\QQ)=0$, $j=1,2, 3$ and $H^4(R_\QQ)$ and $\H^8(R_\QQ)$ are non-zero.
The characteristic class of $p\co \bE_\QQ\to R_\QQ$ is
$[p]=x$, and of the T-dual bundle
is $[p^\#]=y$. The space $\bE_\QQ$ is homotopy equivalent
to $K(\bZ,7)_\QQ\times K(\bZ,4)_\QQ$, so its cohomology ring is
\[            
\H^*(\bE_\QQ) =\QQ[y,\iota_7]/(\iota_7^2)     
\]     
where $y$ is in degree $4$ {\lp}pulled back 
from generators of the same name in $\H^4(R_\QQ)${\rp}, $\iota_7$ is 
the canonical generator of $\H^7(K(\bZ,7)_\QQ)$ in degree $7$.
\end{lemma}

\begin{proof} From the proof of Theorem \ref{thm:classspace}, $\bE$ is the
homotopy fiber of the map $c\co R\to BG$, and can be identified with 
$\Maps(G, K(\ZZ, 7))$, which splits as a product of $K(\bZ,7)$ and
$\Maps^+(G, K(\ZZ, 7))$. Thus $\bE$ has
the homotopy type of $
K(\bZ,7)\times K(\bZ, 4)$. 
Now the result follows by considering the Serre spectral
sequences for the fibrations
\[
\begin{aligned}
K(\bZ,7)_\QQ \to &\bE_\QQ \to K(\bZ, 4)_\QQ, \\
K(\bZ,7)_\QQ \to &R_\QQ \to K(\bZ^{2}, 4)_\QQ,
\end{aligned}
\]
as in \cite{MR2}.
\end{proof}

\subsection{The meaning of the rationalized spherical T-duality classifying space} \label{interpsez}

Theorem \ref{thm:geomTduality} says that in higher dimensions, 
although spherical T-duality doesn't work as nicely as spherical
T-duality when the base $M$ has dimension 4, it works nicely {\em rationally}.
Our goal is to explain this result in this section.

In dimensions higher than 4, principal $\sfSU(2)$-bundles are not just classified by cohomology.
More precisely, $\pi_{i+1}(B\sfSU(2)) = \pi_i(\sfSU(2))$ and since $\pi_i(\sfSU(2)) = \pi_i(S^3)$ is in 
general very complicated, the second Chern class $c_2$ of a principal $\sfSU(2)$-bundle is 
usually not a complete invariant of the bundle. 
For instance, $\sfSU(2)$ is a subgroup of $\sfSU(3)$ with quotient $\sfSU(3)/\sfSU(2) \cong S^5$. 
Thus $\sfSU(3)$ is the total space of a principal $\sfSU(2)$-bundle over $S^5$ which cannot be the trivial
bundle since $0=\pi_4(\sfSU(3)) \ne \pi_4(\sfSU(2)\times S^5) = \pi_4(\sfSU(2))=\ZZ_2$.
[For the homotopy groups of $\sfSU(3)$, cf.~\cite{MT}.] Therefore 
cohomology of $S^5$ does not classify principal $\sfSU(2)$-bundles over it. 
This can also be seen by the known result that $[S^5, B\sfSU(2)] = [S^4, \Omega_eB\sfSU(2)]
= [S^4, \sfSU(2)] =\ZZ_2$. In String Theory, the relevant base dimension is 7, and there are partial
results in \cite{CG}, where there is more information about the map defined by the second Chern class
\beq
c_2 : \text{Bun}_{\sfSU(2)}(M) \to \H^4(M;\ZZ)
\eeq
and its failure to be both surjective and injective in general.  However, at least when $M_7$ is a
2-connected rational homology 7-sphere, then the second Chern class $c_2$ is onto and there
is a complete classification in terms of the 2nd Chern class, the $t$-invariants and Eells-Kuiper invariants
(both of which are torsion).

But what does universal rational spherical T-duality (Theorem \ref{thm:geomTduality}) mean? 
Let $(P, h)$ be a pair where $P$ is a principal $\sfSU(2)$-bundle over $M$ and $h:P\to K(\Z,7)$. 
Then $[(P, h)] \in R$. Let $r: R \to R_\QQ$ denote the rationalization map.
Then  $r([(P, h)]) = [(P_\QQ, h_\QQ)] \in R_\QQ$. Recall that by Theorem \ref{thm:geomTduality},  there is a spherical T-duality
involution $T_\QQ$ on $R_\QQ$.
Consider the ``rational spherical T-dual'' $T_\QQ(r([(P, h)])) = T_\QQ[(P_\QQ, h_\QQ)] = [(Q_\QQ, \widehat h_\QQ)] \in R_\QQ$.
Then $(P, h)$ is spherical T-dualizable if and only if $[(Q_\QQ, \widehat h_\QQ)] \in \text{Image}(r)$. The number of spherical T-duals is 
the size of the fiber $r^{-1}([(Q_\QQ, \widehat h_\QQ)])$.

This rational isomorphism can be understood constructively using the result of Ref.~\cite{Granja} that for each dimension $d$ there exists a natural number $N(d)$ such that any multiple of $N(d)$ in $\H^4(M,\ZZ)$ of a $d$-manifold $M$ is $c_2$ of an $\sfSU(2)$ bundle.  Recall that pairs $(P,H)$ are classified by (homotopy classes of) maps $f:M\rightarrow R$ and pairs of pairs $(P,H)$ and $(\widehat P,\widehat H)$ by maps $\tilde{f}:M\rightarrow S$.  Now the homotopy classes of the maps are determined by the $(d+1)$-skeletons of $R$ and $S$.  

We have seen that there is a map from $g:S\rightarrow R$ induced from the map from $B\sfSU(2)$ to $K(\Z,4)$ corresponding to the generator of $\H^4(B\sfSU(2))$ but the above result implies that there is also a map $\tilde{g}$ from the $d+1$ skeleton of $R$ to that of $S$ which acts of $\H^4(M,\ZZ)$ via multiplication by $N(d)$.  This is an isomorphism on $\H^4(M)$ evaluated over the rationals, but is not surjective over the integers when $N(d)>1$.  However, rationally it is sufficient to construct a T-dual pair of pairs as $\tilde{f}=\tilde g f$ is a map $M\rightarrow S$ and the T-dual can be pulled back from $S$.  This will only be a T-dual rationally, as $c_2(\hat P)=N(d)\, \pi_*H$.

\subsection{When the Base $M$ is Dimension $4$} \label{qsez}

The base $M$ (by assumption) and fiber $S^3$ are orientable, connected manifolds.  
Therefore the total space of $\pi: P\rightarrow M$ is orientable and so
\beq
\H^0(M)\cong\H^0(P)\cong\H^4(M)\cong\H^7(P)=\Z. \label{iso}
\eeq
In what follows we will fix these isomorphisms, which implies that we have chosen an orientation.  
The Gysin sequence \eqref{eqn:eqBAa} implies an isomorphism $\H^7(P;\ZZ)\stackrel{\pi_*}{\cong}\H^4(M;\ZZ)$, which 
can be used to construct a 
7-class $H=(\pi_*)^{-1}c_2(P)$ given any choice of second Chern class $c_2(P)\in\H^4(P)$.  
As there exists a  principal  $\sfSU(2)$-bundle $P\rightarrow M$ for any second Chern class $c_2(P)\in\H^4(P)$, 
a choice of 7-cocycle $H\in\H^7(P)$  also can be used to construct a principal $\sfSU(2)$-bundle $P$ characterized by $c_2(P)=\pi_*(H)$.

The spherical T-duality map $(c_2(P),H)\rightarrow (c_2(\widehat{P}),\widehat{H})$ is the direct sum of two 
such isomorphisms.  First, given $H\in\H^7(P)$ one constructs the principal  $\sfSU(2)$-bundle $\widehat\pi:\widehat{P}\rightarrow M$ 
with second Chern class $c_2(\widehat{P})=\pi_* H$. Then one uses the isomorphism 
$\widehat\pi_*:\H^7(\widehat P)\rightarrow\H^4(M)$ to construct $\widehat{H}=(\widehat\pi_*)^{-1}c_2(P)\in\H^7(\widehat{P})$.  
Summarizing, the spherical T-duality map is
\beq
(\pi_*,(\widehat\pi_*)^{-1}):(\H^4(M),\H^7(P))\rightarrow(\H^4(M),\H^7(\widehat P)):(c_2(P),H)
\mapsto (\pi_*(H),(\widehat\pi_*)^{-1}c_2(P)). \label{t}
\eeq
Thus we have learned that in the special case in which $M$ is 4-dimensional, spherical T-duality 
behaves like ordinary T-duality in the sense that a choice of principal $\sfSU(2)$-bundle $P$ and an integral 7-class on $P$ 
uniquely (up to isomorphism) determines a T-dual principal $\sfSU(2)$-bundle $\widehat P$ and an 
integral 7-class $\widehat H$ on $\widehat P$.

\medskip


\section{Examples}

\subsection{Bundles over $S^4$} \label{sferasez}

When $M$ is a 4-manifold the Gysin sequence and the $H$-twist 
only affect the top and bottom cohomologies of $M$.  As $M$ is orientable and, without loss of generality, connected, 
these are both isomorphic to $\ZZ$.  Therefore calculations of the cohomology of $P$, $\widehat P$, and
the correspondence space $P\times_M \widehat P$ via the Gysin sequence are essentially the same for all 
connected 4-manifolds, as the cohomology in middle dimensions plays no role.  

In this section we will calculate these groups and maps for the case $M=S^4$.  As the maps 
described here act trivially on the middle cohomology, the generalization to other orientable 4-manifolds is straightforward.
One can find explicit constructions of principal $\sfSU(2)$-bundles on $S^4$, and also on 
other compact oriented 4-manifolds, in Section 10.6 of \cite{Taubes}. Constructions of principal $\sfSU(2)$-bundles
in higher dimensions via quaternionic divisors can be found in Section 1.c of \cite{CG}. 

Recall that on oriented 4-manifolds $M$ there is a 1-1 correspondence between the principal 
$\sfSU(2)$-bundles and $\H^4(M;\ZZ)$ given by the second Chern class $c_2$.  In the present case 
$\H^4(S^4)\cong\Z$ and so principal $\sfSU(2)$-bundles will be classified by a single integer.  
We will define the bundles $\pi:P\rightarrow S^4$ and $\widehat\pi:\widehat P\rightarrow S^4$ by
\beq
c_2(P)=k\hsp c_2(\widehat P)=j \,.
\eeq
Ordinarily $P$, and therefore $k$, is part of the initial data and $\widehat P$, and therefore $j$, 
are derived by the T-duality map.  The case in which either $j$ or $k$ is zero is rather simple, 
so for concreteness we will consider $j\neq 0$ and $k \neq 0$.

The Gysin sequence \eqref{eqn:eqBAa} easily yields the cohomology of $P$ and $\widehat P$
\beq
\H^0(P)\cong\H^0(\widehat P)\cong\H^7(P)\cong\H^7(\widehat P)\cong\Z\hsp
\H^4(P)\cong\Z_k\hsp\H^4(\widehat P)\cong\Z_j.
\eeq 
In particular, we determine $\H^4(P)$ using the sequence
\beq
\xymatrix{
 \H^0(M) \ar[d]^\cong \ar[r]^{\cup c_2}& \H^4(M)\ar[d]^\cong\ar[r]^{\pi^*} & \H^4(P)\ar[d]^\cong\ar[r]^{\pi_*} & \H^1(M) \ar[d]^\cong \\
 \ZZ \ar[r]^{\times k} & \ZZ \ar[r] & ? \ar[r] & 0 
}
\eeq  
from which we learn that $\pi^*:\H^4(M)\rightarrow\H^4(P):k\mapsto 0$ and so\ $\H^4(P)\cong \ZZ_k$, and similarly for $\widehat P$.

The other interesting part of the sequence is
\beq
\xymatrix{
\H^7(M) \ar[d]^\cong  \ar[r]^{\pi^*} & \H^7(P) \ar[d]^{\cong}\ar[r]^{\pi_*} & \H^{4}(M) \ar[d]^\cong\ar[r]^{\cup c_2} & \H^8(M) \ar[d]^\cong \\
0 \ar[r] & \ZZ  \ar[r] & \ZZ  \ar[r]^{\times k} & 0
}
\eeq  
which implies that $\pi_*:\H^7(P)\rightarrow\H^4(M)$ is an isomorphism and similarly for $\widehat\pi$.  
It is this isomorphism which allows $H\in\H^7(P)$ to uniquely determine $\widehat c_2\in\H^4(M)$ and $c_2\in\H^4(M)$ 
to uniquely determine $\widehat H\in\H^4(\widehat P)$ and so render the spherical T-dual existent and unique.

Recall that we also impose the condition
\beq
p^*(\widehat H) - \widehat p^*(H) = 0 \in\H^7(P\times_M \widehat P)  \,. \label{corrcond}
\eeq
The pushforward condition has already determined $\widehat H$ uniquely, so does it satisfy (\ref{corrcond})?  
By Theorem \ref{thm:thBAb} (i) we know that it must.  However we will check this directly by calculating $\H^7(P\times_M \widehat P)$. 

The desired cohomology group appears in two distinct short exact Gysin subsequences
\beq
\xymatrix{
\H^3(P)\ar[d]^\cong \ar[r] & \H^7(P)\ar[d]^\cong \ar[r]^{\widehat p^*\quad} & \H^7(P\times_M \widehat P) \ar[d] \ar[r]^{\quad\widehat p_*} & 
\H^{4}(P)\ar[d]^\cong\ar[r] & \H^8(P) \ar[d]^\cong \\
0 \ar[r]^{\times j}  & \ZZ \ar[r] & ? \ar[r] & \ZZ_k  \ar[r]^{\times j} & 0
} \label{pseq}
\eeq  
and
\beq
\xymatrix{
\H^3(\widehat P) \ar[d]^\cong \ar[r] & \H^7(\widehat P)\ar[d]^\cong \ar[r]^{p^*\quad} & \H^7(P\times_M \widehat P) \ar[d] 
\ar[r]^{\quad p_*} & \H^{4}(\widehat P)\ar[d]^\cong\ar[r]& \H^8(\widehat P) \ar[d]^\cong \\
0 \ar[r]^{\times k}  & \ZZ \ar[r] & ? \ar[r] & \ZZ_j  \ar[r]^{\times k} & 0
}
\eeq  
The first sequence implies that $\H^7(P\times_M \widehat P)$ is an extension of $\ZZ$ by $\ZZ_k$ and the 
second that it is an extension of $\ZZ$ by $\ZZ_j$.  This is only possible if
\beq
\H^7(P\times_M \widehat P)\cong\Z\oplus\Z_i\hsp \frac{k}{i},\frac{j}{i}\in\Z.
\eeq
However this is not enough information to determine the group uniquely.

Define $a\in\Z$ and $b\in\Z_i$ by
\beq
\xymatrix{
\H^7(P) \ar[d]^\cong \ar[r]^{\widehat p^*\quad} & \H^7(P\times_M \widehat  P) \ar[d]^\cong \\
\ZZ \ar[r] & \ZZ\oplus\ZZ_i \\
1 \ar[r] & (a,b) }   
\eeq
By exactness of (\ref{pseq}), $\widehat p_*:\H^7(P\times_M \widehat  P)\rightarrow\H^4(P)=\Z_k$ is surjective therefore
\beq
\Z_k\cong\im(\widehat p_*)|_{\H^7(P\times_M \widehat P)}\cong\frac{\H^7(P\times_M \widehat P)}{\im(\widehat p^*)|_{\H^7(P)}}
\cong\frac{\Z\oplus\Z_i}{(a,b)\Z}\cong\Z_{\frac{ai}{\gcd(b,i)}}\oplus\Z_{\gcd(b,i)} \label{zk}
\eeq
where $\gcd(b,i)=i$ if $b=0$ mod $i$.  The total order of the right hand side must be $k$ and so $a=k/i$.  The last two terms on the right hand side combine into a single cyclic group only if their degrees are relatively prime $\gcd(\gcd(b,i),\frac{k}{\gcd(b,i)})=1.$   An identical procedure for $\widehat P$ implies
\beq
p^*:\H^7(\widehat P)\cong\Z\rightarrow\H^7(P\times_M \widehat  P)\cong\Z\oplus\Z_i:1\mapsto(j/i,\hat b) \,, \label{pex}
\eeq
where $\gcd(\gcd(\widehat b,i),\frac{j}{\gcd(\hat b,i)})=1.$ 
 
We will now use the fact that
\beq
\widehat\pi^*\pi_*=p_*\widehat p^*:\H^7(P)\cong\Z\rightarrow\H^4(\widehat P)\cong\Z_j. \label{comm}
\eeq
As $\pi_*:\H^7(P)\rightarrow\H^4(S^4)$ is an isomorphism and $\widehat\pi^*:\H^4(S^4)\rightarrow\H^4(\widehat P)$ is surjective, $\widehat\pi^*\pi_*$ maps the generator, $1$, of $\H^7(P)\cong\Z$ to an order $j$ element of $\H^4(\widehat P)\cong\Z_j$.   On the other hand we have seen that 
\beq
\widehat p^*(1)=(k/i, b)\in\H^7(P\times_M \widehat P)\cong\Z\oplus\Z_i. \label{phatex}
\eeq
The kernel of the map $p_*:\H^7(P\times_M \widehat P)\cong\Z\oplus\Z_i\rightarrow\H^4(\widehat P)\cong\Z_j$ is the image of $p^*:\H^7(\widehat P)\rightarrow\H^7(P\times_M \widehat P)$ which is generated by $(j/i,\hat b)$.  Therefore $p_*\widehat p^*(1)$ will be of order $j$ if $\widehat p^*(1)=(k/i,b)$ is of order $j$ in $(\Z\oplus\Z_i)/\langle j/i,\hat b\rangle$.

Clearly the order of $(k/i,b)$ is at most $j$, because
\beq
j\left(\frac{k}{i},b\right)=\left(\frac{jk}{i},\hat a bi\right)=\left(\frac{jk}{i},a\hat b i\right)=k\left(\frac{j}{i},\hat b\right)
\eeq
in $\Z\oplus\Z_i$.  Let $n=\gcd(j,k)$  As $i$ is also a divisor of $j$ and $k$, $n/i$ is an integer. Now
\beq
\frac{ji}{n}\left(\frac{k}{i},b\right)=\left(\frac{ijk}{n},0\right)=\frac{ki}{n}\left(\frac{j}{n},\hat b\right)
\eeq
and so $\widehat p^*(1)=(k/i,b)$ is of order $ji/n$.  But it must be of order $j$ in order for the commutation condition (\ref{comm}) to be satisfied.  Therefore $i=n$ and so we have computed
\beq
\H^7(P\times_M \widehat P)\cong\Z\oplus\Z_{\gcd(j,k)} \label{sette}
\eeq
and we have found
\beq
a=\frac{k}{i}=\frac{k}{\gcd(j,k)}\hsp \widehat a=\frac{j}{\gcd(j,k)}.
\eeq
It is not hard to see that Eq.~\eqref{sette} applies to any oriented 4-manifold $M$ with vanishing first Betti number.  Furthermore, this result can be generalized to any oriented 4-manifold $M$ by taking the direct sum with $\Z^{{\rm b}_1(M)}$.

Finally we are ready to determine the pullbacks of the twists to the correspondence space $P\times_M \widehat P$.  
First, notice that as $H$ is a multiple of $k$ and $\widehat H$ is a multiple of $j$, they both vanish modulo $\gcd(j,k)$ 
and so their pullback to the second term in (\ref{sette}) will be equal to zero.  The first term is given by $a$ and $\widehat a$
\beq
p^*\widehat H=(\widehat a k,\widehat b k)=\left(\frac{jk}{\gcd(j,k)},0\right)=(a j,b j)=\widehat p^* H.
\eeq
Therefore the value of $\widehat H$ determined by the condition $\widehat\pi_* H=c_2(P)$ indeed agrees with $H$ when both are lifted to the correspondence space, as is required for the consistency of the T-duality map.

In Subsec.~\ref{setsez} we will need the rest of the cohomology of the correspondence space and the related maps.  At other dimensions the extension problems have unique solutions and so the cohomology groups can be directly read from the Gysin sequence.   At dimension 3
\beq
\xymatrix{
\H^3(\widehat P)\ar[d]^\cong \ar[r]^{p^*} & \H^3(P\times_M \widehat P) \ar[d] \ar[r]^{p_*} & \H^{0}(\widehat P)\ar[d]^\cong\ar[r]^{} 
& \H^{4}(\widehat P)\ar[d]^\cong\\
0 \ar[r] & ? \ar[r] & \ZZ \ar[r]^{\times k} & \ZZ_j 
}
\eeq  
Therefore $p_*$ is an injection into $\H^{0}(\widehat P)\cong\Z$.   Furthermore it yields an isomorphism
\beq
\H^3(P\times_M \widehat P)\cong\ker(k\wedge)\cong\gcd(j,k)\Z\cong\Z
\eeq
and $p_*: \H^3(P\times_M \widehat P)\rightarrow \H^{0}(\widehat P):1\mapsto\gcd(j,k)$.

At dimension 10 the Gysin sequence yields an isomorphism
\beq
\H^{10}(P\times_M \widehat P)\cong\H^7(\widehat P) \cong \ZZ\,.
\eeq
Finally, at dimension 4 the Gysin sequence yields 
\beq
\xymatrix{
\H^0(\widehat P)\ar[d]^\cong \ar[r] & \H^4(\widehat P)\ar[d]^\cong \ar[r]^{p^*} & \H^4(P\times_M \widehat P)\ar[d] \ar[r]^{p_*} & 
\H^{1}(\widehat P)\ar[d]^\cong  \\
\ZZ \ar[r]^{\times k} & \ZZ_j \ar[r]  & ? \ar[r] & 0 
}
\eeq  
and so
\beq
\H^4(P\times_M \widehat P)\cong\frac{\Z_j}{\im(k\wedge:\Z\rightarrow\Z_j)}\cong \frac{Z_j}{k\Z_j}\cong\Z_{\gcd(j,k)} \,. 
\eeq
We could have calculated $\H^7$ more quickly using the fact that its torsion part is isomorphic to that of $\H^4$ by the universal coefficient theorem.  However this longer derivation has provided explicit expressions (\ref{pex}) and (\ref{phatex}) for the pullback maps which will be used in Subsec.~\ref{setsez}.  

\subsection{Bundles over $\HH P^n$}

In \cite{GS, CG}, the complete classification of principal $\sfSU(2)$-bundles 
over the quaternionic projective spaces $\HH P^2$, and $\HH P^3$, has been given. 
In both cases, they are classified by a precisely described subset of integers and precisely
described torsion (homotopy) groups.  Suppose that $n>1$ and that $c_2(P)\neq 0\in\H^4(\HH P^n)\cong \Z$,  
then by the Gysin sequence one finds $\H^7(P)=0$.  Here we have used the fact that the 
cohomology groups $\H^p(\HH P^n) \cong \Z$ if $p=4j, \, j=0,1, \ldots, n$, and $\H^p(\HH P^n)=0$ otherwise
(in particular  $\H^7(\HH P^n)=0$). 
Therefore, in this case $H=0$ will be trivial in cohomology and so the spherical T-dual $\widehat P\rightarrow M$ 
will have $c_2(\widehat P)=0$, which in the case $n= 2$ implies that it is the trivial bundle $\widehat P=\HH P^2\times S^3$.  
By the K\"unneth theorem $\H^7(\HH P^2\times S^3)\cong\Z$ and so the dual twist $\widehat H$ will characterized by a 
single integer.   However not all integers values of $\widehat H$ will be realized by T-duality because in these cases, the 2nd 
Chern class $c_2(P)$ is not onto $\H^4(\HH P^2)\cong\ZZ$, so only certain values of $H$ are T-dualizable. 
More precisely, $c_2$ has to be of the form $24 r +s\in\H^4(\HH P^2)$ where $r\in \ZZ$ and $s=0, 1, 9, 16$.  
If $s=0$ or $s=16$ then the bundle is uniquely defined, if $s=1$ or $s=9$ then there are two distinct bundles 
with the same $c_2$.  In either case, the spherical T-dual is $\widehat P=\HH P^2\times S^3$ with $\widehat H=24r+s$.  
However, this T-duality is only injective and so invertible if $s$ is even.

In the case $M=\HH P^3$, again only certain values of $c_2$ are allowed but now $c_2$ never completely classifies the bundle, 
even when $c_2=0$.  Thus even though $\H^7(P)=0$ implies that $H=0$, this only implies that $c_2(\widehat P)=0$ and does not 
determine whether the T-dual bundle is trivial or not.  Nonetheless the isomorphism $\H^7(\widehat P)\cong\Z\cong\H^4(\HH P^3)$ 
can be used to determine $\widehat H\in\H^7(\widehat P)$.  Again the fact that only certain elements of $\H^4(\HH P^3)$ are realized as 
$c_2$ of principal $\sfSU(2)$-bundles (cf.\ \cite{GS, CG}) implies that only certain pairs $(\widehat P,\widehat H)$ will be T-dualizable, 
and in no case will their T-dual be unique.

\medskip


\section{Spherical T-duality is an Isomorphism of Twisted Cohomology} \label{isosez}

In this section we will see that, for a base of any dimension, T-duality induces an isomorphism on twisted cohomologies with real or rational coefficients.  In this section it will be implied that all cohomology groups and twisted cohomology groups are over the real numbers.

\subsection{Construction of the T-dual Twist as a Differential Form}

The Gysin sequence
\beq
\xymatrix{
\cdots \ar[r]^{c_2\wedge}& \H^7(M)\ar[r]^{\pi^*} & \H^7(P)\ar[r]^{\pi_*} & \H^{4}(M)\ar[r]^{c_2\wedge} & \cdots 
}
\eeq  
gives a decomposition of the differential 7-form $H\in\Omega^7(P)$ in terms of closed forms $H_4\in\Omega^4(M)$ and $H_7\in\Omega^7(M)$ as
\beq
H=\pi^* H_7 + (\pi_*)^{-1} H_4
\eeq
where without loss of generality we have not included an additional exact term $dB$ which can be eliminated via an automorphism of the twisted cohomology corresponding to multiplication by $e^{-B}$.   The inverse of $\pi_*$ is well-defined as a map on closed forms
\beq
 (\pi_*)^{-1}:\ker(c_2\wedge)|_{\Omega^4(M)}\rightarrow \frac{\Omega^7(P)}{\pi^*\Omega^7(M)}.
\eeq
However, by choosing an element in each coset it lifts to $\Omega^7(P)$ itself 
\beq
 (\pi_*)^{-1}:\ker(c_2\wedge)|_{\Omega^4(M)}\rightarrow \Omega^7(P):\alpha\mapsto\alpha\wedge \CS(A).
\eeq
Therefore, as was seen using Cartan-Weil theory in Subsec. \ref{cwsez}, $H$ may be decomposed as
\beq
H=\pi^* H_7 + \widehat c_2\wedge \CS(A). \label{hforma}
\eeq
Similarly
\beq
\widehat H=\widehat\pi^* H_7 + c_2\wedge \CS(\widehat A). \label{hhatforma}
\eeq
We have seen that in general $\widehat H$ may contain an additional summand $c_2\wedge \pi^* a$ where $a$ is a closed 3-form on $M$.  However, as $d(\pi^*a)=0$ and $p^*(c_2\wedge \pi^*a)=d(\CS(A)\wedge \pi^*a)$, the $a$ may be absorbed as a shift $ \CS(\widehat A)\to \CS(\widehat A)+\pi^*a$.  This redefinition will be implied in all formulas in Sec. \ref{isosez}.   The resulting shift in Eq. (\ref{tomega}) does not affect Eq. (\ref{eta}), as each term in (\ref{tomega}) needs to vanish separately.  Thus the choice of $a$ will not affect the injection proof.    In the surjection proof, in addition to the shift in $\CS(\widehat A)$, we will see that $a$ also affects Eq. (\ref{chiusocond}).  As the isomorphism proof below holds for any choice of $a$, a change of $a$ will induce an isomorphism between the corresponding $d-\widehat H$ cohomologies of $\widehat P$.

Note that, with these choices made, $H$ and $\widehat H$ are automatically $\sfSU(2)$-invariant.  Furthermore
\beq
dH=\pi^* dH_7+\widehat c_2\wedge c_2=0\hsp d \widehat H=\widehat \pi^* dH_7+c_2\wedge \widehat c_2=0
\eeq
as $H_7$ is closed and $c_2\wedge \widehat c_2=0$.

Lifting to the correspondence space $P\times_{M} \widehat P$ one finds
\beq
p^* \widehat H-\widehat p^* H = p^* \widehat\pi^* H_7 - \widehat p^*\pi^* H_7 +    c_2\wedge \CS(\widehat A) - \widehat c_2\wedge \CS(A).
= d(\CS(A)\wedge \CS(\widehat A)).
\eeq

\subsection{Spherical T-duality Induces an Isomorphism}

Given closed, $\sfSU(2)$-invariant 7-forms $H\in\Omega^7(P)$ and $\widehat H\in\Omega^7(\widehat P)$ we define the $H$ $(\widehat H)$ twisted cohomology $\H_H^{\rm{even/odd}}(P)$\ $(\H_{\widehat H}^{\rm{even/odd}}(\widehat P))$ to be the subset of $\Omega^{\rm{even/odd}}(P)$ ( $\Omega^{\rm{even/odd}}(\widehat P)$) which is annihilated by the operation $d_H=d-H\wedge$ ($d_{\widehat H}=d-\widehat H\wedge$) quotiented by the image of the same operation on $\Omega^{\rm{odd/even}}(P)$ ( $\Omega^{\rm{odd/even}}(\widehat P)$).

Let $\omega$ be an $\sfSU(2)$-invariant $d_H$ closed polyform representing a class in $\H_H^{\rm{even/odd}}(P)$.  Note that every twisted cohomology class has a representative which is $\sfSU(2)$-invariant, we will restrict our attention to these representatives and furthermore will choose representatives $H$ and $\widehat H$ of the form Eq.~\eqref{hforma} and \eqref{hhatforma}. Lifting to the correspondence space 
$P\times_{M} \widehat P$, applying the kernel $\exp(\CS(A)\wedge \CS(\widehat A))$ and integrating over the fiber, we define the 
{\em T-duality transform}
 \beq
T_*(\omega) = \int_{\sfSU(2)} \exp(\CS(A)\wedge \CS(\widehat A)) \wedge \widehat p^*\omega.
 \eeq
\begin{lemma} \label{homlem}
The T-duality transform $T_*$ induces a homomorphism of twisted cohomology groups
 \beq
T_* : \H_H^{\rm{even/odd}}(P) \longrightarrow \H_{\widehat H}^{\rm{odd/even}}(\widehat P).
\eeq
\end{lemma}
\begin{proof}
Since $d(\CS(A)\wedge \CS(\widehat A))=-\widehat p^*H+ p^*{\widehat H}$, we have
\bea
T_*(d_H\omega)&=&\int_{\sfSU(2)} \exp(\CS(A)\wedge \CS(\widehat A)) \wedge \widehat p^*d\omega\nonumber\\&&-\int_{\sfSU(2)} \exp(\CS(A)\wedge \CS(\widehat A)) \wedge \widehat p^*H\wedge\widehat p^*\omega\nonumber\\
&=&-d\int_{\sfSU(2)} \exp(\CS(A)\wedge \CS(\widehat A)) \wedge \widehat p^*\omega\nonumber\\
&&-\int_{\sfSU(2)} \exp(\CS(A)\wedge \CS(\widehat A)) \wedge (\widehat p^*H-\widehat p^*H+ p^*{\widehat H})\wedge\widehat p^*\omega\nonumber\\
&=&-d_{\widehat{H}}T_*(\omega) \label{tcomm}
\eea
where in the last step we used the fact that $\int_{\sfSU(2)}= p_*$ together with the property of the pullback
\beq
p_*(\alpha\wedge p^*\beta)=(p_*(\alpha))\wedge\beta.
\eeq
Eq. (\ref{tcomm}) may be summarized by the statement $T\circ d_{H}=-d_{{\widehat H}}\circ T$.
Therefore $T$ takes $d_H$ exact (closed) forms on $P$ to $d_{\widehat H}$ exact (closed) forms on $\widehat P$ and so it induces a well-defined homomorphism on the twisted cohomology groups.  
\end{proof}

As the Maurer-Cartan forms $\theta^k$ are a basis of left-invariant forms and the connection $A^k$ restricted to the fiber yields $\theta^k$, for any $\omega\in\Omega^\bullet(P)^{\sfSU(2)}$,
we have the decomposition
\beq
\omega =\pi^*\omega_0+\CS(A)\wedge\pi^*\omega_3+\sum_k A^k\wedge\pi^*\omega_1^k+\sum_{j,k} A^j\wedge A^k\wedge\pi^*\omega_2^{jk} \label{omegadec}
\eeq
where $\omega_\bullet\in\Omega^\bullet(M)$.  The T-dual of $\omega$ is then
\bea
T_*(\omega)&=&\int_{\sfSU(2)} \CS(A)\wedge\left(\pi^*\omega_3+
\CS(\widehat A)\wedge\pi^*(\omega_0)\right)\nonumber\\
&&+\frac{1}{24\pi^2}\int_{\sfSU(2)}\sum_{i,j,k} \epsilon^{ijk}A^i\wedge A^j\wedge A^k\wedge F^i \wedge\pi^*\omega_2^{jk}\nonumber\\
&=&\pi^*\omega_3+\CS(\widehat A)\wedge\pi^*(\omega_0-\sum_{i,j,k}\epsilon^{ijk}F^i \wedge\pi^*\omega_2^{jk}) \label{tomega}
\eea 
where $F^i$ is the Lie-algebra valued curvature corresponding to the direction $i$ in the Lie algebra.  So the kernel of $T_*$ consists of all invariant polyforms $\eta$ of the form 
\beq
\eta =\sum_{i,j,k}\epsilon^{ijk}F^i \wedge\pi^*\eta_2^{jk}+ \sum_k A^k\wedge\pi^*\eta_1^k+\sum_{j,k} A^j\wedge A^k\wedge\pi^*\eta_2^{jk}. \label{eta}
\eeq

The map  $T_* : \H_H^{\rm{even/odd}}(P) \longrightarrow \H_{\widehat H}^{\rm{odd/even}}(\widehat P)$ is injective if and only if all  $d_H$-closed forms in $\ker(T_*)$ are $d_H$-exact.  In other words, to show that $T_*$ is an injection we must show that $\eta$ is only $d_H$ closed if it is also $d_H$ exact.

By the Maurer-Cartan equation
\beq
d\theta^i=\epsilon^{ijk}\theta^j\wedge\theta^k\hsp dA^i=\epsilon^{ijk}A^j\wedge A^k+ F^i.
\eeq

Suppressing the various $\pi^*$,
\bea
d_H\eta
&=&d_H\sum_{i,j,k}\epsilon^{ijk}F^i \wedge\eta_2^{jk}+\sum_k F^k\wedge\eta_1^k\\
&+&\sum_i A^i\wedge (-d_{H_7}\eta_1^i-2\sum_kF^k\wedge\eta_2^{ik})+\sum_{j,k}A^j\wedge A^k\wedge(\sum_i \epsilon^{ijk}\eta_1^i+d_{H_7} \eta_2^{jk})\nonumber.
\eea
As the monomials in $A^i$ are linearly independent, each term must vanish separately.  In particular, the first line contains only constant and cubic terms and so must vanish separately from the second line.  Therefore $\eta$ is only $d_H$ closed if 
\beq
\eta_1^i=-\sum_{i,j,k}\epsilon^{ijk}d_{H_7}\eta_2^{jk} \label{hchiuso}
\eeq
which implies that $d_{H_7}\eta_1^i=0$ and so $\sum_k\eta_2^{ik}\wedge F^k=0$ for all $i$.  In this case
\beq
\eta
=d_H\sum_{i,j,k}\epsilon^{ijk}(A^i\wedge\eta_2^{jk})
\eeq
and so $\eta$ is also $d_H$ exact.  Therefore $T_* : \H_H^{\rm{even/odd}}(P) \longrightarrow \H_{\widehat H}^{\rm{odd/even}}(\widehat P)$ is injective.

What about surjectivity?   Any $\widehat\sfSU(2)$-invariant form on $\widehat P$ can be written in the form
\beq
\widehat\omega =\widehat\pi^*\widehat\omega_0+\CS(\widehat A)\wedge\widehat\pi^*\widehat\omega_3+\sum_k \widehat A^k\wedge\widehat\pi^*\widehat\omega_1^k+\sum_{j,k}\widehat A^j\wedge\widehat A^k\wedge\widehat\pi^*\widehat\omega_2^{jk}.\label{omegahat}
\eeq
Surjectivity on twisted cohomology is the claim that if $d_{\widehat H}\widehat\omega=0$ then $\widehat\omega$ can be written in the form (\ref{tomega}) plus an exact form $d_{\widehat H}\eta$.

The space of invariant differential forms can be split into two subspaces, $\Lambda^*(\widehat P)=\Lambda^A(\widehat P)\oplus\Lambda^B(\widehat P)$. $\Lambda^A(\widehat P)$ consists of sums $\widehat \pi^*\widehat\omega_0+\pi^*\widehat\omega_3\wedge \CS(\widehat A)$ while $\Lambda^B(\widehat P)$ consists of sums $\sum_k\widehat \pi^*\widehat\omega_1^k\wedge\widehat A^k+\sum_{ij}\widehat \pi^*\widehat\omega_2^{ij}\wedge\widehat A^i\wedge\widehat A^j$.  The key observation is again that $d_{\widehat H} \Lambda^A(\widehat P)\subset\Lambda^A(\widehat P).$  Therefore, if $\widehat\omega$ is $d_{\widehat H}$ closed then its restriction to $\Lambda^B(\widehat P)$ must also be $d_{\widehat H}$ closed.  But this restriction to  $\Lambda^B(\widehat P)$ is just the restriction to forms of the form (\ref{eta}).  Repeating the above calculation leading to Eq. (\ref{hchiuso}) with $H$ replaced by $\widehat H$ given in Eq. (\ref{hhatforma}), these are only twisted closed if 
\beq
\widehat\omega_1^i=-\sum_{i,j,k}\epsilon^{ijk}d_{(H_7+c_2\wedge a)}\widehat\omega_2^{jk}. \label{chiusocond}
\eeq

Now we can prove surjectivity directly.  For any $\widehat\sfSU(2)$-invariant, $d_{\widehat H}$-closed polyform $\widehat\omega\in\Lambda^*(\widehat P)$, using the decomposition (\ref{omegahat}) one obtains
\beq
\widehat\omega=T_*\omega+d_{\widehat H}\eta \label{ch}
\eeq
where
\beq
\omega=\pi^*\widehat\omega_3+\CS(A)\wedge\pi^*(\widehat\omega_0-\sum_{i,j,k}\epsilon^{ijk}\widehat F^i \wedge\pi^*\widehat\omega_2^{jk})
\eeq
and
\beq
\eta=\sum_{i,j,k}\epsilon^{ijk}(\widehat A^i\wedge\widehat\pi^*\widehat\omega_2^{jk}).
\eeq
So $T_*$ is surjective on $\widehat H$ twisted cohomology.

As $T_*$ is an injective and surjective homomorphism on twisted cohomology, we have proved our main theorem: 

\begin{theorem} \label{princ}
The T-duality transform
 \beq
T_* : \H_H^{\rm{even/odd}}(P) \longrightarrow \H_{\widehat H}^{\rm{odd/even}}(\widehat P)
\eeq
is an isomorphism of twisted cohomology groups.
\end{theorem}

An alternate proof of the theorem is as follows.   $d_{H} \Lambda^A(P)\subset\Lambda^A(P)$ and so $(\Lambda^A(P),d_H)$ is a differential complex.   A slight modification of (\ref{ch}) yields, for any $\sfSU(2)$-invariant $\omega\in\Lambda^*(P)$ decomposed as in Eq. (\ref{omegadec}), an element $\tilde\omega\in\Lambda^A(P)$
\beq
\tilde\omega=\omega-d_H\left(\sum_{i,j,k}\epsilon^{ijk}(A^i\wedge\pi^*\omega_2^{jk})\right).
\eeq
As $\tilde\omega-\omega$ is $d_H$ exact, this map is a chain homotopy.  Therefore the cohomology of the complex $(\Lambda^A(P),d_H)$ is isomorphic to that of all invariant differential forms, which in turn is isomorphic to that of all differential forms.  

Now T-duality satisfies $T\circ d_{H}=-d_{{\widehat H}}\circ T$ and so its action on this complex generates an homomorphism between the twisted cohomologies.  Furthermore $T_*$ acts on $\Lambda^A(P)$ as an involution $\omega_0\longleftrightarrow\omega_3$ and so squares to the identity. Therefore it induces an isomorphism on the twisted cohomology.  As this involution exchanges the even and odd degrees, so does $T_*$.

\medskip

\section{Spherical T-duality with integer coefficients}

In various cases the objects and maps defined above, with rational coefficients, have integral lifts.  In this section all cohomology groups will be over integer coefficients.

\subsection{The base $M$ a 4-manifold}

In this subsection we will demonstrate that, when the base is an oriented 4-manifold, the spherical T-duality map can be extended to cohomology with integral coefficients.  We have seen that it is a pairing between the second Chern class, which is an integral 4-cocycle in the base, and the twist, which is an integral 7-cocycle in the total space of the bundle.   We will demonstrate that spherical T-duality induces an isomorphism between the twisted integral cohomologies of the respective $\sfSU(2)$ principal bundles.

\subsubsection{Isomorphism of twisted cohomologies over the integers}

We will now show

\begin{theorem} \label{evenprop} For $P\rightarrow M$ and $\widehat P\rightarrow M$ principal $\sfSU(2)$-bundles 
over an oriented 4-manifold related by the map (\ref{t}), there is an isomorphism between the twisted cohomology 
groups $\H^{\rm{even/odd}}_H(P)$ and $\H^{\rm{odd/even}}_{\widehat{H}}(\widehat P)$, defined to be the cohomology 
of the integral cohomology groups with respect to the cup product with the corresponding twist
\beq
\H_H^k(P)=\frac{\ker(H\cup)|_{\H^k(P)}}{\im(H\cup)|_{\H^{k-7}(P)}}.
\eeq
\end{theorem}
\begin{proof}
If $c_2(P)=0$ then the K\"unneth theorem yields 
\beq
\H^k(P)\cong\H^k(M)\oplus\H^{k-3}(M)
\eeq
and so
\bea
&&\H^1(P)\cong\H^1(M)\hsp \H^2(P)\cong\H^5(P)\cong\H^2(M)\hsp \H^3(P)\cong \H^3(M)\oplus\Z\nonumber\\&& 
\H^4(P)\cong\H^1(M)\oplus\Z\hsp \H^6(P)\cong \H^3(M).
\eea
As $H\cup$ annihilates all cohomology classes except for $\H^0$, these classes are isomorphic to the corresponding twisted classes.  The only classes which are different in the twisted case are
\beq
\H^0_H(P)\cong 0 \hsp \H^7_H(P)\cong \Z_j \label{twisted}
\eeq
where $H=j\in\H^7(P)$.  Note that if $j=0$ then $\H^0_H(P)=\H^7_H(P)=\Z$.

In the case $c_2(P)\neq 0$ we will instead compute these cohomology groups using the Gysin sequence
\beq
\xymatrix{
\cdots \ar[r]^{c_2\cup}& \H^k(M)\ar[r]^{\pi^*} & \H^k(P)\ar[r]^{\pi_*} & \H^{k-3}(M)\ar[r]^{c_2\cup} & \cdots 
}
\eeq 
As the image of $c_2\cup$ is trivial at all degrees except for $c_2\cup:\H^0(M)\rightarrow\H^4(M)$, this long exact sequence yields a short exact sequence for each of $k$ except for 3 and 4.

In particular, for $k\leq 2$, $\H^{k-3}(M)=0$ and so the Gysin sequence yields
\begin{equation*}
\xymatrix{
 \H^{k-4}(M) \ar[d]^\cong  \ar[r] & \H^k(M)\ar[r]^{\pi^*} & \H^k(P)\ar[r]^{\pi_*} & \H^{k-3}(M) \ar[d]^\cong   \\
0 & & & 0 }
\end{equation*}
therefore
\beq
\H^1(P)\cong \H^1(M)\hsp \H^2(P)\cong\H^2(M). \label{eq12}
\eeq
Similarly, for $k\geq 5$, $\H^k(M)=0$ and so the sequence reduces to
\begin{equation*}
\xymatrix{
\H^{k}(M) \ar[d]^\cong  \ar[r]^{\pi^*} & \H^k(P)\ar[r]^{\pi_*}  & \H^{k-3}(M)\ar[r]^{c_2\cup}  & \H^{k+1}(M) \ar[d]^\cong   \\
0 & & & 0 }
\end{equation*}
and so 
\beq
\H^5(P)\cong \H^2(M)\hsp \H^6(P)\cong\H^3(M). \label{eq56}
\eeq

Let $c_2=i\neq 0\in\H^4(M)$.  Then $c_2\cup:\H^0(M)\rightarrow\H^4(M)$ is injective, at $k=3$ one finds
\begin{equation*} \xymatrix{
0 \ar[r] & \H^3(M) \ar[r]^{\pi^*}  & \H^3(P)  \ar[r]^{\pi_*\qquad} & \ker(c_2\cup)|_{\H^{0}(M)}=0 }
\end{equation*} 
and so 
\beq
\H^3(P)\cong \H^3(M) \,.
\eeq
Finally, since the image of $c_2\cup:\H^0(M)\rightarrow\H^4(M)$ is $i\Z\subset\Z\cong\H^4(M)$, at $k=4$ one obtains
\beq
\xymatrix{
0 \ar[r] & \frac{\H^4(M)}{\im(c_2\cup)|_{\H^0(M)}} = \Z_i \ar[r]^{\qquad\pi^*} & \H^4(P) \ar[r]^{\pi_*} & \H^{1}(M)\ar[r]^{c_2\cup} &  0
}
\eeq
Since $\H^1(M)$ is free by the universal coefficient theorem, all finite order elements of $\H^4(P)$ must be in $\ker(\pi_*)$ and so $\im(\pi^*)$.  Therefore the short exact sequence splits into finite and infinite order elements mapped by $\pi^*$ and $\pi_*$ respectively, yielding
\beq
\H^4(P)=\H^1(M)\oplus\Z_i.
\eeq

As in the case $i=0$, the twisted and untwisted cohomologies are isomorphic except at degrees where 0 and 7 where they are given by Eq. (\ref{twisted}) or in the untwisted case $j=0$ are both isomorphic to $\Z$.

The T-duality map (\ref{t}) acts by simply exchanging the integers $i$ and $j$.  In the case $i=j=0$ in which $P\cong M\times S^3$, $\widehat P=\widehat{S}^3$ and $H=\widehat H=0$ the total even and odd twisted cohomologies are
\bea
\H^{\rm{even}}_{H=0}(M\times S^3)\cong\H^{\rm{even}}_{\widehat H=0}(M\times \widehat{S}^3)&\cong&\H^1(M)\oplus\H^2(M)\oplus\H^3(M)\oplus\Z^2\\
\H^{\rm{odd}}_{H=0}(M\times S^3)\cong\H^{\rm{odd}}_{\widehat H=0}(M\times\widehat{S}^3)&\cong&\H^1(M)\oplus\H^2(M)\oplus\H^3(M)\oplus\Z^2\nonumber
\eea
and so, as expected from Theorem \ref{evenprop}, 
\beq
\H^{\rm{even/odd}}_H(P)\cong\H^{\rm{odd/even}}_{\widehat{H}}(\widehat P). 
\eeq

If $i\neq 0$ but $j=0$ then $\widehat P\cong M\times \widehat{S}^3$ and $H=0$.  In this case, assembling the above results one finds
\bea
\H^{\rm{even}}_{H=0}(P)\cong\H^{\rm{odd}}_{\widehat H=i}(M\times \widehat{S}^3)&\cong&\H^1(M)\oplus\H^2(M)\oplus\H^3(M)\oplus\Z\oplus\Z_i\\
\H^{\rm{odd}}_{H=0}(P)\cong\H^{\rm{even}}_{\widehat H=i}(M\times \widehat{S}^3)&\cong&\H^1(M)\oplus\H^2(M)\oplus\H^3(M)\oplus\Z\nonumber
\eea
again satisfying Theorem \ref{evenprop}.  The case $i=0$ and $j\neq 0$ proceeds identically, with even and odd degrees interchanged.

The last case is $i\neq 0$ and $j\neq 0$.  Now both $P$ and $\widehat{P}$ are nontrivial bundles and neither $H=j$ nor $\widehat{H}=i$ vanishes.  In this case, assembling the above results one finds
\bea
\H^{\rm{even}}_{H=j}(P)\cong\H^{\rm{odd}}_{\widehat H=i}(\widehat P)&\cong&\H^1(M)\oplus\H^2(M)\oplus\H^3(M)\oplus\Z_i\\
\H^{\rm{odd}}_{H=j}(P)\cong\H^{\rm{even}}_{\widehat H=i}(\widehat P)&\cong&\H^1(M)\oplus\H^2(M)\oplus\H^3(M)\oplus\Z_j\nonumber
\eea
completing the proof of Theorem \ref{evenprop}.
\end{proof}

In fact, Theorem \ref{evenprop} is a corollary to the following theorem, which can similarly be derived from the above formulae for $\H(P)$.
\begin{theorem} \label{seiprop}
 For $P\rightarrow M$ and $\widehat P\rightarrow M$ principal $\sfSU(2)$ bundles over an oriented 4-manifold related by the map (\ref{t}) and for all integers $k$,
\beq
\bigoplus_j \H^{6j+k}_H(P)\cong\bigoplus_j \H^{6j+k+3}_{\widehat{H}}(\widehat P). \label{seieq}
\eeq
\end{theorem}

To derive Theorem \ref{evenprop} from Theorem \ref{seiprop} one need only take the direct sum over the values $k=0,\ 2$ and $4$.  

\subsubsection {Isomorphism of twisted K-theories }

As $P$ is an oriented 7-manifold, this isomorphism extends to the corresponding 7-twisted K-theories up to an extension problem which we will ignore in what follows.  

To define K-theory on $P$, twisted by a closed 7-form $H_7$ representing $k$ times the 
generator of $\H^7(P,\ZZ)$, we first recall from Corollary 4.7 in \cite{DP} that the generator of 
$\H^7(S^7,\ZZ)$ corresponds to the Dixmier-Douady invariant of an algebra bundle 
$\cE\to S^7$ with fibre a stabilized infinite Cuntz $C^*$-algebra $O_{\infty} \otimes \cK$. 
Now let $f:P \to S^7$ be a degree $k$ continuous map, then $f^*(\cE) \to P$ is an algebra 
bundle with fibre a stabilized infinite Cuntz $C^*$-algebra $O_{\infty} \otimes \cK$ and 
Dixmier-Douady invariant equal to $k$ times the generator of $\H^7(P,\ZZ)$. 
Then, by \cite{DP2}, the twisted K-theory is defined as $\K^*(P, H_7) = \K_*(C_0(P, f^*(\cE)))$, 
where $C_0(P, f^*(\cE))$ denotes continuous sections of $f^*(\cE)$ vanishing at infinity. 
This shows that $\K^*(P, H_7)$ is well defined, although we will not use the explicit construction.

The $H$-twisted K-theory of an oriented seven manifold $P$ can be computed using a two step spectral sequence with differentials $d_3=Sq^3$ and $d_7=H\cup$.  The second differential may be derived from a Mayer-Vietoris argument similar to that presented for the calculation of K-theory twisted by a 3-cocycle in Ref.~\cite{BS}.  This may be compared with the one step spectral sequence used above to construct twisted cohomology, which only used the differential $d_7$.

The operation $Sq^3$ annihilates all cocycles of degree less than 3.  Also, as $P$ is of dimension 7 and $Sq^3$ increases the degree of a cocycle by 3, it annihilates all cocycles of degree greater than 4.  The image of $Sq^3$ is a $\Z_2$ torsion element of integral cohomology, but $\H^7(P)=\Z$ and so has no $\Z_2$ torsion therefore $Sq^3$ also annihilates $\H^4(P)$.   Finally, the operation acts on any element of $\H^3(P)$ by squaring it.  More precisely, such elements can be decomposed using the Gysin sequence
\beq
\xymatrix{
0 \ar[r] & \H^3(M) \ar[r]^{\pi^*} & \H^3(P) \ar[r]^{\pi_*} & \H^{0}(M)\ar[r]^{\quad c_2(P)\cup} & } \label {gy3}
\eeq 
Now we will consider two cases.  First, if $c_2(P)=0$ then the last map is the zero map and by the K\"unneth theorem
\beq
\H^3(P)\cong\H^3(M)\otimes\H^0(S^3)\oplus \H^0(M)\otimes\H^3(S^3).
\eeq
The operation $\text{Sq}^3$ annihilates both $\H^3(M)$ and $\H^3(S^3)$ and so in this case it annihilates $\H^3(P)$.  Second, if $c_2(P)\neq0$ then the last map of (\ref{gy3}) is an injection and so $\pi_*\H^3(P)=0$ therefore any element $a\in\H^3(P)$ is the pullback $\pi^*b$ of an element of $b\in\H^3(M)$.  As $\text{Sq}^3$ is natural and annihilates $\H^3(M)$
\beq
\text{Sq}^3a=\text{Sq}^3(\pi^*b)=\pi^*\text{Sq}^3b=\pi^*0=0.
\eeq
Therefore $\text{Sq}^3\H^*(P)=0$.  

As the kernel of $\text{Sq}^3$ is $\H^*(P)$ and the image is trivial, the first step of the spectral sequence does not affect the cohomology of $P$.  The second step, the cohomology with respect to $d_7=H\cup$, is identical to the only step in the spectral sequence for the computation of twisted cohomology.  
Thus we have proved 
\begin{theorem} \label{kprop} 
For $\pi:P\rightarrow M$ a principal $\sfSU(2)$ bundle over an oriented 4-manifold with second Chern class $c_2(P)$ 
and $H\in\H^7(P)$ such that $\pi_*H=c_2(P)$, up to an extension problem there is an isomorphism between the twisted 
cohomology $\H^{\rm{even/odd}}_H(P)$ and the $H$-twisted K-theory $\K^{\rm{even/odd}}_H(P)$ .
\end{theorem}

Combining Theorems \ref{evenprop} and \ref{kprop} one obtains

\begin{theorem}  For $P\rightarrow M$ and $\widehat P\rightarrow M$ principal $\sfSU(2)$ bundles over an oriented 4-manifold related by the map (\ref{t}), there is an isomorphism between the twisted K-theories $\K^{\rm{even/odd}}_H(P)$ and $\K^{\rm{odd/even}}_{\widehat{H}}(\widehat P)$.
\end{theorem}

\subsubsection{Example}
Consider the 4-dimensional sphere $S^4$ as the projective quaternionic space $\HH P^1$. 
Define a quaternionic line bundle $\cL$ over $\HH P^1$ as follows. Consider the trivial rank 2 quaternionic vector bundle 
$\HH^2 \times \HH P^1 \to \HH P^1$, and define $\cL$ as the subbundle $\{(x, L) \in \HH^2 \times \HH P^1: 
x\in L\}$. It is known that the sphere bundle $\sfSU(2)\to S(\cL) \to \HH P^1$ has second Chern class 
equal to $-1$.  Therefore the pair $(S({\cL}),  {\rm vol})$ has spherical T-dual the pair $(S^7, -1\cdot {\rm vol})$ consisting of the Hopf bundle $S^7\to S^4$.



\subsection{Special Case: $\widehat P$ is the trivial bundle}



If $H=0$ so that $c_2(\widehat P)=0$ then $\widehat H=\widehat\pi^*c_2\cup a$ where $\widehat\pi_*a=1$.  Up to an extension problem the Gysin sequence splits to yield
\bea
\H^k_{H=0}(P)&\cong&\ker(c_2\cup)|_{\H^{k-3}(M)}\oplus\frac{\H^k(M)}{\im(c_2\cup)}\label{arbdim}\\
\H^k_{\widehat H}(M\times \widehat S^3)&\cong&\ker(c_2\cup)|_{\H^{k}(M)}\oplus\frac{\H^{k-3}(M)}{\im(c_2\cup)}
\nonumber
\eea
and so the isomorphism (\ref{seieq}) of Theorem \ref{seiprop} extends to manifolds $M$ of arbitrary dimension in the case $H=0$. 

When $H\in\H^7(P)$ is the pullback $\pi^* h$ of a class $h\in\H^7(M)$ then
\beq
c_2(\widehat P)=\pi_*H=\pi_*\pi^*h=0
\eeq
where the last equality follows from $\pi_*\pi^*=0$ which is implied by the exactness of the Gysin sequence.  In this case, $\widehat H=c_2\cup a+\widehat\pi^* h$ and Eq. (\ref{arbdim}) continues to hold, but using the $h$-twisted cohomology of $M$ and so Theorems \ref{seiprop} and \ref{evenprop} extend to this case as well.

\subsection{The Dimension of $M$ is Less than or Equal to 7}

In this subsection we will compute the integral twisted cohomology groups explicitly in the case in which the dimension of the base $M$ is less than or equal to 7 and we will see that the even twisted cohomology of $P$ continues to be isomorphic to the T-dual odd twisted cohomology of $\widehat P$.   Twisted cohomology can be computed from a spectral sequence beginning with ordinary cohomology whose first differential $d_1=H\cup$ and whose second differential is a dimension 13 secondary operation.  As the total spaces of $P$ and $\widehat P$ are at most of dimension 10, only the first differential $d_1$ acts nontrivially and so, up to an extension problem the twisted cohomology is just the cohomology with respect to $d_1=H\cup$.  It is this cohomology with respect to $d_1$ which we will compute and, since it is anyway only equal to $d_H$ cohomology up to an extension problem, we will assume that all exact sequences split so that we only compute the $d_1$ cohomology itself up to another extension problem.

 Note that twisted cohomology and $d_1$ cohomology are both ill-defined in general as $d_1\cup d_1=H\cup H$ only vanishes mod 2 and so the differential is not nilpotent.  However, again since we are only interested in manifolds of dimension less than or equal to 10, these classes vanish and so integral twisted cohomology is well-defined.

The Gysin sequence
\beq
\stackrel{c_2\cup}{\rightarrow}\H^k(M)\stackrel{\pi^*}{\rightarrow}\H^k(P)\stackrel{\pi_*}{\rightarrow}\H^{k-3}(M)\stackrel{c_2\cup}{\rightarrow}
\eeq 
may be reduced to the short exact sequence
\beq
0{\rightarrow}\frac{\H^k(M)}{\im(c_2\cup)}\stackrel{\pi^*}{\rightarrow}\H^k(P)\stackrel{\pi_*}{\rightarrow}\ker(c_2\cup)|_{\H^{k-3}(M)}{\rightarrow} 0.
\eeq 
Recall that we are only calculating the twisted cohomology up to an extension problem, and so we may assume that the sequence splits, yielding a decomposition of $\H^k(P)$
\beq
\H^k(P)=\pi^*\left(\frac{\H^k(M)}{\im(c_2\cup)}\right)\oplus(\pi_*)^{-1}\left(\ker(c_2\cup)|_{\H^{k-3}(M)}\right). \label{pco}
\eeq
Note that $\pi_*$ maps surjectively onto $\ker(c_2\cup)$ and so its inverse is well defined modulo an element in the image of $\pi^*$.

We will use this decomposition to decompose $H$ as
\beq
H=\pi^*h+(\pi_*)^{-1}{\widehat c_2}\hsp h\in\H^7(M)\hsp \widehat{c}_2\in\H^4(M).
\eeq
More generally we can decompose an arbitrary element $A\in\H^k(P)$ as
\beq
A=\pi^*a+(\pi_*)^{-1}\tilde a\hsp a\in\H^k(M)\hsp \tilde{a}\in\H^{k-3}(M).
\eeq
The first term in (\ref{pco}) corresponds to $h$ and $a$ while while $\widehat c_2$ and $\tilde a$ correspond to the second.  Therefore  (\ref{pco}) implies that
\beq
c_2\cup\widehat{c}_2=c_2\cup\tilde{a}=0
\eeq
while $h$ and $a$ are only defined modulo $\im(c_2\cup)$.

To calculate the cohomology twisted by $H$, we must determine the action of $H\cup$ on $\H^*(P)$.   Eq. (\ref{pco}) may be used to decompose this product
\beq
H\cup A=\pi^*b+(\pi_*)^{-1}\tilde b\hsp b\in\H^{k+7}(M)\hsp \tilde{b}\in\H^{k+4}(M)
\eeq
where $c_2\cup\tilde b=0$ and $b$ is only defined modulo $\im(c_2\cup)$.   In particular, $c^k$ is in the kernel of $H\cup$, and so represents an $H$-twisted cohomology class, if and only if both $b\in\im(c_2\cup)$ and $\tilde b=0$.

Using the properties of the pushforward and pullback maps
\beq
f_*f^*=0\hsp f_*(\alpha\cup f^*\beta)=(f_*\alpha)\cup\beta
\eeq
we can calculate
\bea
\tilde b&=&\pi_*(H\cup A)=\pi_*\left((\pi^*h+(\pi_*)^{-1}{\widehat c_2})\cup(\pi^*a+(\pi_*)^{-1}\tilde a)\right)\nonumber\\
&=&\pi_*\pi^*(h\cup a)+h\cup\pi_*(\pi_*)^{-1}\tilde a+\pi_*(\pi_*)^{-1}\widehat c_2\cup a+\pi_*\left((\pi_*)^{-1}\widehat c_2\cup(\pi_*)^{-1}\tilde a \right)\nonumber\\&=&h\cup\tilde a+\widehat c_2\cup a.
\eea
Here we used $(\pi_*)^{-1}\widehat c_2\cup(\pi_*)^{-1}\tilde a=0$  which is a result of the fact that the product of two 7 cocycles would be a 14 cocycle, but the total space $P$ is at most 10-dimensional, it must vanish.



Although $a$ is only defined modulo $\im(c_2\cup)$, $\widehat c_2\cup c_2=0$ and so $\widehat c_2\cup a$ and therefore $\tilde b$ is well-defined.

Similarly, this splitting can be used to compute $b$  by assembling the remaining terms
\beq
\pi^*b=\pi^*h\cup\pi^*a=\pi^*(h\cup a).
\eeq
As the kernel of $\pi^*$ is $\im(c_2\cup)$, this yields
\beq
b=h\cup a
\eeq
where it is understood that $b$ is only defined modulo $\im(c_2\cup)$.

Now we are ready to write a representative of general $H$-twisted cocycle in $\H^k(P)$.  It is an element of $\ker(H\cup)$, and so an element of $\H^k(P)$ such that $\tilde b=b=0$
\beq
\{\pi^*a+(\pi_*)^{-1}\tilde a|a\in\H^k(M),\tilde a\in\H^{k-3}(M),h\cup \tilde a+\widehat c_2\cup a=c_2\cup\tilde a=0,h\cup a\in\im(c_2\cup)\}.
\eeq
A twisted cocycle $\H^k_H(P)$ consists of such elements where one quotients $a$ by $\im(c_2\cup)$, as $a$ is only defined up to this equivalence, and also one quotients by $\im(H\cup)$ which decomposes into quotients of the components $a$ and $\tilde a$ by the values of $\pi^* b$ and $(\pi_*)^{-1}\tilde b$ respectively
\beq
\H^k_H(P)=\frac{\{a\in\H^k(M),\tilde a\in\H^{k-3}(M)|h\cup \tilde a+\widehat c_2\cup a=c_2\cup\tilde a=0,h\cup a\in\im(c_2\cup)\}}{a\sim a+\im(c_2\cup), \tilde a\sim\tilde a+\im(h\cup\ker(c_2\cup)),(a,\tilde a)\sim(a,\tilde a)+\im(h\cup,\widehat c_2\cup)}. \label{hh}
\eeq

If $h=0$ then the twisted cohomology simplifies to
\beq
\H^k_H(P)=\frac{\{a\in\H^k(M),\tilde a\in\H^{k-3}(M)|\widehat c_2\cup a=c_2\cup\tilde a=0\}}{a\sim a+\im(c_2\cup), \tilde a\sim\tilde a+\im(\widehat c_2\cup)}. \label{senzah}
\eeq
Under spherical T-duality, $c_2\leftrightarrow\widehat c_2$ which leaves the conditions and relations of (\ref{senzah}) invariant if $a\leftrightarrow\tilde a$, which shifts the degree of each generator by $3$, extending Theorem \ref{seiprop} to arbitrary dimensions in the case in which $h=0\in\H^7(M)$.  

Note that $h=0$ automatically if the dimension of $M$ is less than or equal to six.  Furthermore, if $M$ is an orientable 7-manifold and $h\neq 0$ then $h\cup\widehat a+\widehat c_2\cup a=0$ can be reduced to $\widehat c_2\cup a^\prime=0$ by shifting $a\rightarrow a^\prime=a+\gamma$ at degree $k=3$, where $\widehat c_2\cup\gamma=h$.  Note that $\gamma$ exists because $h$ is proportional to the top form and $M$ is orientable.  At other degrees $h\cup\tilde a=0$ for dimensional reasons and so again the condition reduces to $\widehat c_2\cup a=0$.  

Now T-duality, exchanging $c_2\leftrightarrow\widehat c_2$ and $a^\prime\leftrightarrow\tilde a$, again leaves the twisted cohomology invariant but shifts the degree of each generator by 3, in accordance with Theorem \ref{seiprop}.  
Thus we have extended Theorem \ref{seiprop} and so also its corollary Theorem \ref{evenprop} to the case in which the dimension of $M$ is less than or equal to 7.

\subsection{Example: Bundles over $S^4\times S^3$} \label{setsez}

We have seen that when the dimension of the base $M$ is equal to seven, spherical T-duality is complicated by the fact that part of the twist may arise from the pullback of a 7-class on the base.  This 7-class prevents a unique choice of $\widehat H$ already when M is a 7-manifold and when the dimension is higher than 7 it prevents us from proving that spherical T-duality induces an isomorphism on integral twisted cohomology, which is well defined when the dimension of $P$ is less than or equal to 13, corresponding to dim$(M)=10$.

In this subsection we will consider the example $M=S^4\times S^3$, in which the richness of the 7-dimensional case can be seen.  In fact, at dim$(M)=7$, the obstructions described above only occur when $\H^4(M)$ contains nontorsion classes, as in this case.  Let $\alpha$ and $\beta$ be the generators of $\H^4(M)$ and $\H^3(M)$ respectively.  Define the $\sfSU(2)$ principal bundles $P$ and $\widehat P$ to be Cartesian products of $S^3$ with the bundles in Subsec. (\ref{sferasez}).  In particular
\beq
c_2(P)=k\alpha\hsp c_2(\widehat P)=j\alpha.
\eeq
Notice that $P$ and $\widehat P$ are not the only bundles with their Chern classes, one may also cut out a 7-ball and reglue it using a transition function representing a nontrivial element of $\pi_6(S^3)=\Z_{12}$.  These other bundles $\widehat P$ are therefore also T-dual to $(P,H)$.  Such choices would not change the Gysin sequences and so calculations of cohomology groups below, but will obstruct the constructions of bundle automorphisms that we will then describe.

The cohomology of $P$ is easily obtained from the K\"unneth theorem together with the cohomology of an $\sfSU(2)$-bundle over $S^4$
\beq
\H^0(P)\cong\H^3(P)\cong\H^{10}(P)\cong\Z\hsp \H^4(P)\cong\Z_k\hsp \H^7(P)\cong\Z\oplus\Z_k
\eeq
and similarly for $\widehat P$ with $j\leftrightarrow k$.  It is the degree 7 cohomology which is relevant for spherical T-duality, which is described by the Gysin sequence
\beq
\xymatrix{
\H^3(M)\ar[d]^\cong \ar[r]^{\cup c_2} & \H^7(M) \ar[d]^\cong \ar[r]^{\pi^*} & \H^7(P) \ar[d]^\cong\ar[r]^{\pi_*} & \H^{4}(M) \ar[d]^\cong\ar[r]^{\cup c_2} & \H^8(M)
\ar[d]^\cong  \\
\ZZ \ar[r]^{\times k} & \ZZ \ar[r] & \ZZ\oplus\ZZ_k \ar[r] & \ZZ \ar[r]  & 0 
}
\eeq  
As $\pi_*$ is surjective it must act as $\pi_*:(m\in\Z,n\in\Z_k)\mapsto m$.  Therefore $c_2(\widehat P)=\pi_*(H)$ is independent of $n\in\Z_k$.  Similarly, $c_2(P)=\widehat\pi_*(\widehat H)$ will not determine the $\Z_j$ torsion part of $\widehat H$.  However these torsion parts are still restricted by the condition $p^*\widehat H=\widehat p^* H\in\H^7(\widehat P_M\times P)$ to which we will now turn.  First however we comment that the bundle automorphisms that we will describe later determine the images of the maps $\pi^*:\H^7(M)\rightarrow\H^7(P)$ and  $\widehat \pi^*:\H^7(M)\rightarrow\H^7(\widehat P).$ 

The seventh cohomology group of the correspondence space can be found from that in Subsec. \ref{sferasez} together with the K\"unneth theorem
\beq
\H^7(P\times_M \widehat P)\cong (\Z\oplus\Z_{\gcd(j,k)})\otimes \H^0(S^3)\oplus \Z_{\gcd(j,k)}\otimes \H^3(S^3)\cong\Z\oplus\Z_{\gcd(j,k)}\oplus\Z_{\gcd(j,k)}.
\eeq
We will need the pullbacks to the correspondence space.  First consider $\widehat p^*H$.  This can be calculated using the Gysin subsequence
\beq
\xymatrix{
\H^3(P) \ar[r]  \ar[d]^\cong&  H^7(P) \ar[d]^\cong \ar[r]^{\widehat p^*\quad} & \H^7(P\times_M \widehat P)  
\ar[d]^\cong \ar[r]^{\quad \widehat p_*}  & \H^4(P) \ar[d]^{\cong} \ar[r] & \H^8(P) \ar[d]^\cong \\
\ZZ \ar[r]^{(0,\times j)} & \ZZ\oplus \ZZ_k \ar[r]  & \ZZ\oplus\ZZ_{\gcd(j,k)}^2 \ar[r] & \ZZ_k \ar[r]^{\times j} & 0 
}
\eeq
The maps can be constructed by applying the K\"unneth theorem to the corresponding maps in the case $M=S^4$. 
\bea
&&\widehat p ^*:\H^7(P)\longrightarrow \H^7(P\times_M \widehat P):(m,n)\rightarrow\left(\frac{km}{\gcd(j,k)},mb,n\right)\nonumber\\
&&p^*:\H^7(\widehat P)\longrightarrow \H^7(P\times_M \widehat P):(\widehat m,\widehat n)\rightarrow\left(\frac{j\widehat m}{\gcd(j,k)},\widehat m\widehat b,\widehat n\right).
\eea

Recalling that the relation $c_2(\widehat P)=\pi_*H$ implies $m=j$ and similarly $\widehat m=k$, these maps are
\bea
\widehat p^* H&=&\widehat p^*(j,n)=\left(\frac{kj}{\gcd(j,k)},jb,n\right)=\left(\frac{kj}{\gcd(j,k)},0,n\right)\in\Z\oplus\Z_{\gcd(j,k)}^2\nonumber\\
p^* \widehat H&=&\widehat p^*(k,\widehat n)=\left(\frac{jk}{\gcd(j,k)},kb,\widehat n\right)=\left(\frac{jk}{\gcd(j,k)},0,\widehat n\right)\in\Z\oplus\Z_{\gcd(j,k)}^2.
\eea
Note that the first two components of $\widehat p^*H$ and $p^*\widehat H$ agree and the condition that the third are equal is equivalent to 
\beq
n=\widehat n\in\Z_{\gcd(j,k)}.
\eeq
This condition does not uniquely determine $\widehat n$ given $n$, the ambiguity is an element of $\Z_{j/\gcd(j,k)}$.   And so we see that indeed the pair $(P,H)$ does not uniquely determine $(\widehat P,\widehat H)$ as expected when dim$(M)\geq 7$.

Also in the case of the topological T-duality for circle bundles, $\widehat H$ is not determined uniquely from the original pair $(P,H)$.  However it is determined up to a bundle automorphism.  In the case at hand $\widehat H$ is also determined up to a bundle automorphism.  Construct a map $g:M=S^4\times S^3\rightarrow\sfSU(2)$ by composing any projection $M\rightarrow S^3$ which is the identity on the $S^3$ with an identification between $S^3$ and the group manifold $\sfSU(2)$.     The bundle isomorphism is just the fiberwise right multiplication of $\widehat P$ by $g(m)$ for each $m\in M$, $\widehat n-n$ mod $\gcd(j,k)$ times.

Does this choice affect the T-duality isomorphism between the integral twisted cohomologies?  The nonvanishing elements of the integral twisted cohomology groups are
\bea
&&\H_H^4(P)\cong\Z_k\hsp\H^{10}_H(P)\cong\Z_j\hsp\H^7_H(P)\cong\frac{\Z\oplus\Z_j}{(j,n)\Z}\cong\Z_{\frac{jk}{\gcd(k,n)}}\oplus\Z_{\gcd(k,n)}\nonumber\\
&&\H_{\widehat H}^4(\widehat P)\cong\Z_k\hsp\H_{\widehat H}^{10}(\widehat P)\cong\Z_j\hsp\H_{\widehat H}^{7}(\widehat P)\cong\frac{\Z\oplus\Z_j}{(k,n)\Z}\cong\Z_{\frac{jk}{\gcd(j,n)}}\oplus\Z_{\gcd(j,n)}.
\eea
When $n=0$ the even and odd twisted cohomology groups are all isomorphic to $\Z_j\oplus\Z_k$ and so T-duality gives a true isomorphism.  However more generally T-duality relates two distinct extensions of $\Z_j$ by $\Z_k$, and so provides an isomorphism up to an extension problem.   This should be of no surprise both as the demonstration of the isomorphism the previous subsection was only performed up to an extension problem and also the identification of the twisted cohomology with the $H$ cohomology via the spectral sequence only holds up to an extension problem.

\medskip


\section{Applications to supergravity and string theory}

\subsection{10-dimensional Supergravity} \label{stringsez}

The data in 10-dimensional supergravity includes a Lorentzian 10-manifold $Y^{10}$ together with several real valued $(p+1)$-chains called $p$-branes and $p$-cochains called $p$-fluxes.   The chains and cochains must satisfy certain consistency conditions called the equations of motion.  We will restrict our attention to the most common case of interest in which there exists a diffeomorphism $Y^{10}\cong \R\times X^9$ which induces a foliation of $Y^{10}$ into copies $X^9_t$ of $X^9$ with $t\in\R$ such that the metric induced on each $X^9_t$ is Riemannian.   An allowed $(p+1)$-chain on $Y^{10}$ can be intersected with $X^9_t$ to yield a $p$-chain on $X^9_t$ corresponding to a given $p$-brane.  At the level of cohomology the K\"unneth theorem allows a decomposition of a $p$-flux on $Y^{10}$ into a $p$-class and a $(p-1)$-class on $X^9_t$.

A charge on $X^9_t$ is a map $g$ from the $p$-chains corresponding to allowed $p$-branes on $X^9_t$ to a set $S$ .  The charge of the brane is the image of the map applied to the chain corresponding to that brane.  As there are many different choices of maps with many possible sets as their image, the set of charges is too large to be useful.  For example, the set $S$ may just be the set of chains in $X^9_t$.  The set of charges is defined to be the image of $g$, which is a subset of $S$.  If $S$ is the set of chains and the $p$-fluxes are equal to zero, then the set of charges is the set of cycles on $X^9_t$. 

A more useful notion is that of a conserved charge.  This is a charge on $X^9_t$ which, for any allowed $p$-brane, will have an image which is independent of $t$.   The set of conserved charges is the image of the corresponding map $g$.  The set of cycles does not satisfy this condition, because if a $p$-brane whose $(p+1)$-cycle on $Y^{10}$ restricts to a cycle $Z^p$ of $X^9_t$, the equations of motion allow it to restrict to other cycles $\tilde Z^9$ of $X^9_{\tilde t}$.  However, the equations of motion demand, if all of the $p$-fluxes vanish, that $Z^p$ and $\tilde Z^p$ are homologous. Therefore, the homology group $\H_p(X^9)$ is a set of conserved charges for $p$-branes.   As $X^9$ is oriented in IIB supergravity, we may use Poincar\'e duality to identify these conserved charges with the cohomology group $\H^{9-p}(X^9)$.

In this section it will be understood that all cohomology groups are calculated with real coefficients, as branes and fluxes in supergravity correspond to chains and cochains with real coefficients.  We will comment upon extensions to string theory, where Dirac quantization implies that the chains and cochains have integral coefficients.

\subsection{Type IIB Supergravity with Zero Flux}

Each supergravity theory comes with a set of $p$-branes and $p$-fluxes.  We will consider Type IIB supergravity.  In this theory there are D1-branes, D3-branes, D5-branes, D7-branes, F1-branes and NS5-branes.   Therefore a configuration will consist of two 1-chains and 5-chains and one 3-chain and 7-chain on each $X^9_t$.   The $p$-fluxes are $F_1$, $F_3$, $F_5$, $F_7$, $F_9$, $H_3$ and $H_7$.  Note that, for $p\geq 0$, there is a pairing between $p$-branes and $(p+2)$-fluxes.  In string theory this pairing extends to $p=-1$.  The $p$-fluxes must satisfy consistency conditions $F_p=*F_{10-p}$ and $H_7=*H_3$.

Let us begin with the case in which all fluxes vanish.  Now, as we wrote in the previous subsection, $\H^{9-p}(X^9)$ is an allowed group of conserved charges for $p$-branes.  Therefore the corresponding conserved charges for D-branes are the even cohomology $\H^{\rm even}(X^9)$ while NS5-brane and F1-brane charges are $\H^4(X^9)$ and $\H^8(X^9)$ respectively.   These charges are summarized in Table \ref{cariche}.

What about T-duality?  We will first consider ordinary T-duality, in which $X^9$ admits a free circle action whose space of orbits $M^8$, so that $X^9$ is a circle bundle over $M^8$.  The T-dual of $X^9$ will be $\widehat X^9\cong M^8\times S^1$.  The T-dual $H_3$ vanishes if and only if \cite{BEM} $X^9$ is the trivial circle bundle $X^9\cong M^8\times S^1$, for now we will consider this case.  Then by the K\"unneth theorem the even and odd cohomologies of $X^9$ and $\widehat X^9$ are both isomorphic to the cohomology of $M^8$ and so isomorphic to each other.  The T-duality map induces the isomorphism between the even (odd) cohomology of $X^9$ and the odd (even) cohomology of $\widehat X^9$.  

Thus T-duality provides an isomorphism on the D-brane charges.  However T-duality does {\it not} provide an isomorphism on the NS brane charges.  The T-dual of an NS5-brane with respect to a circle action which leaves the corresponding 5-cycles invariant yields another NS5-brane but a more general circle action can lead to a T-dual in which the NS5-brane disappears entirely, being replaced with a degeneration in the dual circle. Therefore T-duality only induces an isomorphism on the set of D-brane charges $\H^{\rm even}(M)$, not on the set of all $p$-brane charges $\H^{\rm even}(M)\oplus\H^4(M)\oplus\H^8(M)$.

\begin{table}[position specifier]
\centering
\begin{tabular}{c|l|l|l|l|}
Brane&No flux&$H_3\neq 0$&$G_7\neq 0$&$H_7\neq 0$\\
\hline\hline
D1&$\H^8(X^9)$&$\H^8_{H_3}(X^9)$&$\H^8(X^9)$&$\H^8_{H_7}(X^9)$\\
\hline
D3&$\H^6(X^9)$&$\H^6_{H_3}(X^9)$&$\H^6_{G_7}(X^9)$&$\H^6_{H_7}(X^9)$\\
\hline
D5&$\H^4(X^9)$&$\H^4_{H_3}(X^9)$&$\H^4_{G_7}(X^9)$&$\H^4(X^9)$\\
\hline
D7&$\H^2(X^9)$&$\H^2_{H_3}(X^9)$&$\H^2_{G_7}(X^9)$&$\H^2(X^9)$\\
\hline
F1&$\H^8(X^9)$&$\H^8(X^9)$&$\H^8_{G_7}(X^9)$&$\H^8(X^9)$\\
\hline
NS5&$\H^4(X^9)$&$\H^4(X^9)$&$\H^4(X^9)$&$\H^4_{H_7}(X^9)$\\
\hline
NS7&&&&$\H^2_{H_7}(X^9)$\\
\hline
\end{tabular}
\caption{Set of conserved charges for various branes in IIB supergravity with all fluxes equal to zero and with all fluxes equal to zero except for one}
\label{cariche}
\end{table}

\subsection{Type IIB Supergravity with $H_3$ Flux}

Now let the cocycle $H_3$ on $Y^{10}$ be nonzero.   The Hodge duality condition $H_7=*H_3$ is also nonzero.  The conserved charges only depend on the cohomology classes of the fluxes.  Using the K\"unneth theorem, we can decompose $\H^p(Y^{10})\cong\H^p(X^9)\oplus\H^{p-1}(X^9)$ and so decompose
\beq
H_3=\tilde H_3 + \tilde H_2\wedge e\hsp H_7=\tilde H_7 + \tilde H_6\wedge e \label{em}
\eeq
where $e$ generates $\H^1(\R)$.   Up to an exact element, the Hodge duality condition then relates $\tilde H_3$ to $\tilde H_6$ and $\tilde H_2$ to $\tilde H_7$.  However, it turns out that $\tilde H_2$ and $\tilde H_6$ do not affect the set of conserved branes.  Therefore only $\tilde H_3$ and $\tilde H_7$ will be relevant here.  These two cocycles of $X^9$ are not related by any Hodge duality and so can be chosen independently.  To keep the notation as uncluttered as possible, we will suppress the tildes.

We will begin by considering the case in which $H_3$ is arbitrary but $H_7$ and all other cocycles vanish.  In this case D-branes corresponding to certain cohomology classes are not allowed by the equations of motion.  The allowed cohomology classes are precisely the kernels of the differential operators in the spectral sequence which determines the twisted cohomology of $X^9$ with respect to $d_H=d-H\wedge$, as was shown over the reals in Ref~\cite{andreas} by adapting the analogous argument over the integers in Ref.~\cite{MMS}, where the spectral sequence yields twisted K-theory \cite{BCMMS}.  Similarly, a D-brane restricted to $X^9_t$ and $X^9_{\tilde t}$ can represent two distinct cohomology classes.  However these two cohomology classes are always equal up to an element of the image of these differential operators.  Therefore twisted cohomology $\H^{\rm even}_{H_3}(X^9)$ is a set of conserved charges for D-branes in type IIB supergravity. 

This set of conserved charges is invariant under T-duality:  It was shown in Ref.~\cite{BEM} that the T-duality map induces an isomorphism between the even (odd) twisted cohomology of $X^9$ and the odd (even) twisted cohomology of its T-dual $\widehat X^9$.  Here $\widehat X^9$ is a circle bundle over $M^8$ with Chern class equal to the pushforward of $H$ by the projection map $\pi:X^9\rightarrow M^8$.  In type IIB string theory $H$ is an integral 3-cocycle and the corresponding conserved charges are given by twisted K-theory, and it was shown that T-duality induces an isomorphism in this setting as well.

Recall that the conserved charges are a map from the set of $p$-branes to a set. Twisted cohomology only classifies D-branes, not NS-branes, and so the kernel of this map includes all F1 and NS5-branes.   As a result this set of conserved charges is not maximal, it only classifies some of the objects in the theory.  

\subsection{Type IIB Supergravity with $F_7$ Flux}

In the previous subsection we considered the case in which the cocycle $H_3$ on $X^9$ is arbitrary while $H_7$ represents  $0\in\H^7(X^9)$.    What if instead $H_3$ is exact and $H_7$ is arbitrary?  Some motivation for the study of this case has appeared in the work of Sati; relations between the roles played by 3 forms and 7 forms in type II string theory have been discussed in Ref. \cite{hishamns5} while \cite{sati09} used $H_7$ twisted cohomology to study fields and charges in  a simplified version of heterotic string theory.   

Type IIB supergravity and string theory admits a $\Z_2$ automorphism called S-duality which exchanges D5 and NS5-branes, D1 and F1-branes and also $F_{p+2}$ and $H_{p+2}$ fluxes for $p=1$ and $5$, while leaving the others invariant.  The action of S-duality on the D7-brane is somewhat more complicated.  The arguments above all admit simple actions of S-duality \cite{uday}.  What if we use S-duality to change the above question somewhat, to study configurations in which the only nontrivial cocycle is $F_7$?

Several string theory compactifications of interest are on the total spaces of $S^3$ bundles and have nontrivial $F_7$ fluxes, such as string theory on AdS$^3\times{\rm T}^4\times S^3$ and AdS$^3\times{\rm K3}\times S^3$.  These examples also have nontrivial $F_3$ fluxes.  Nonetheless, the spherical T-duals are easily computed, they are products of AdS$^3$ with nontrivial $S^3$ bundles fibered over $T^4$ and ${\rm K3}$ respectively.  As the base is a 4-manifold, spherical T-duality and the corresponding isomorphisms on twisted cohomologies are well defined over the integers.  If $H_7$ is equal to $k$ times the top class, then the $H_7$-twisted cohomology of $P=M^4\times S^3$ where $M^4={\rm T}^4$ and K3 are 
\bea
&&\H^0_{H_7}(T^4\times S^3)=0\hsp \H^1_{H_7}(T^4\times S^3)=\H^6_{H_7}(T^4\times S^3)=\Z^4\nonumber\\
&& \H^2_{H_7}(T^4\times S^3)=\H^5_{H_7}(T^4\times S^3)=\Z^6\nonumber\\
&&\H^3_{H_7}(T^4\times S^3)=\H^4_{H_7}(T^4\times S^3)=\Z^5\hsp \H^7_{H_7}(T^4\times S^3)=\Z_k
\eea
and
\bea
&&\H^0_{H_7}(K3\times S^3)=\H^1_{H_7}(K3 \times S^3)=\H^6_{H_7}(K3 \times S^3)=0\nonumber\\
&&\H^2_{H_7}(K3 \times S^3)=\H^5_{H_7}(K3 \times S^3)=\Z^{22}\nonumber\\
&&\H^3_{H_7}(K3 \times S^3)=\H^4_{H_7}(K3 \times S^3)=\Z\hsp \H^7_{H_7}(K3 \times S^3)=\Z_k.
\eea
The spherical T-duals $\widehat P_{T^4}$ and $\widehat P_{K3}$ have $H_7=0$ and so their twisted cohomology is just their ordinary cohomology, which can be calculated from the Gysin sequence
\bea
&&\H^0_{H_7}(\widehat P_{T^4})=\Z\hsp \H^1_{H_7}(\widehat P_{T^4})=\H^6_{H_7}(\widehat P_{T^4})=\Z^4\nonumber\\
&& \H^2_{H_7}(\widehat P_{T^4})=\H^5_{H_7}(\widehat P_{T^4})=\Z^6\nonumber\\
&&\H^3_{H_7}(\widehat P_{T^4})=\Z^4\hsp \H^4_{H_7}(\widehat P_{T^4})=\Z^4\oplus\Z_k\hsp \H^7_{H_7}(T^4\times S^3)=\Z
\eea
and
\bea
&&\H^0_{H_7}(\widehat P_{K3})=\Z\hsp \H^1_{H_7}(\widehat P_{K3})=\H^6_{H_7}(\widehat P_{K3})=0\nonumber\\
&&\H^2_{H_7}(\widehat P_{K3})=\H^5_{H_7}(\widehat P_{K3})=\Z^{22}\nonumber\\
&&\H^3_{H_7}(\widehat P_{K3})=0\hsp \H^4_{H_7}(\widehat P_{K3})=\Z_k\hsp \H^7_{H_7}(\widehat P_{K3})=\Z.
\eea
In the case $M=T^4$ the even and odd conserved charges $\Z^{15}$ and $\Z^{15}\oplus\Z_k$ are exchanged by T-duality as are $\Z^{23}$ and $\Z^{23}\oplus\Z_k$ in the case $M=$K3.

As was described in Ref.~\cite{MMS}, there is a simple 1-1 correspondence between differentials in a spectral sequence which computes the a set of conserved charges and baryonic configurations of the kind introduced in Ref.~\cite{barioni}.  In the case of $F_7$ there is only one kind of baryon, a D7-brane whose cap product with the 7-cocycle $F_7$ is nontrivial and is equal to the boundary of an F1-brane.  This single baryon leads to a single differential $d_1=F_7$ which acts on D7-branes and yields F1-branes.  As a result the set of conserved D7-brane charges is $\ker(d_1)|_{\H^2(X^9)}$ while the conserved F1 charges are $\H^8(X^9)/\im(d_1)$.  A set of conserved charges for all other $p$-branes is just $\H^{9-p}(X^9)$.

This means that we can assemble the F1, D3, D5 and D7 charges together into $F_7$-twisted cohomology $\H^{\rm even}_{F_7}(X^9)$ while the D1 and NS5 are classified by ordinary cohomology $\H^8(X^9)$ and $\H^4(X^9)$.   As $H$ vanishes, the Freed-Witten anomaly in fact implies that, in the integral lift, when $H_7=0$, F1, D3, D5 and D7 charges are classified by untwisted K-theory.  7-twisted K-theory on a 9-manifold, like $X^9$ is just the $d_1$ cohomology of untwisted K-theory and so in fact this subset of branes in type IIB string theory has a set of conserved charges in one to one correspondence with elements of $\K^0_{F_7}(X^9)$.  Unfortunately there is no candidate for a T-dual theory with a 7-flux twist.

What about the original question, what if only $H_7$ is nontrivial?  String theory contains a function $Y^{10}\rightarrow\R$ called the dilaton.  If one allows the dilaton to be defined only locally, then the D7-brane itself has an S-dual, traditionally called the NS7-brane.  This NS7-brane brane appears in some cases of interest \cite{vafa}, indeed in F-theory it can be obtained from a D7-brane by exchanging the two circles in the torus fibration.  Now the corresponding baryon is an NS7-brane with a nontrivial cap product with $H_7$ equal to the boundary of a D1-brane.   The D1, D3, NS5, NS7 conserved charges are $\H^{\rm even}_{H_7}(X^9)$, with an integral lift to $\K^0_{H_7}(X^9)$ in string theory, while F1 and D5 branes are classified by $\H^8(X^9)$ and $\H^4(X^9)$.

Our isomorphism Theorem \ref{princ} states that spherical T-duality will yield an isomorphism to $\H^{\rm even}_{H_7}(\widehat X^9)$.  There is no known prescription however for the T-dual of a configuration with a dilaton which is only locally defined.  Thus the construction of either a T-dual or a spherical T-dual in string theory remains an open problem.  Furthermore, the obvious target of a T-dual, type IIA supergravity or string theory, contains an $H_7$-flux but no obvious conserved charges classified by $H_7$-twisted cohomology.

As a result, unlike ordinary T-duality, spherical T-duality does not correspond to any known isomorphism between physical theories.  However, two spherical T-duals, using distinct $\sfSU(2)$ bundle structures, will again yield the same set of conserved charges $\H_{H_7}^{\rm even}$.  Therefore such a spherical T-duality provide a one to one correspondence between a set of conserved charges in certain distinct IIB supergravity and string compactifications.

\subsection{Classification of Fluxes and Bianchi Identities}

The appearance of $7-twisted$ cohomologies classifying conserved charges of $p$-branes in IIB supergravity can be seen already from the viewpoint of the Bianchi identities of the fluxes, with no branes at all.  In Ref. \cite{moorewitten} the authors noted that there is a 1-1 correspondence between the classifications of $F_p$ fluxes and D-branes given by Stokes' theorem, an observation extended to all fluxes and branes in Ref. \cite{monodromie} and reviewed in \cite{modave}.  In this subsection we will treat the fluxes as differential forms.  Define the $d_{H_3}$-closed field strength
\beq
G_p=F_p+B\wedge F_{p-2} \,,
\eeq
where the exact forms $F_p$, $F_{p-2}$ and $B$ satisfying $H_3=dB$ are defined patchwise on a good cover.  The equations of motion are obtained by setting to zero the variation of the kinetic term
\beq
S\propto\int H_3\wedge H_7+G_3 \wedge G_7 \,,
\eeq
with respect to $B$, where locally $H_3=dB$, to zero.  This yields
\beq
0=\frac{\delta S}{\delta B}= -dH_7+ F_1\wedge F_7\,.
\eeq
One thus sees that the complex $(H_7,F_1)$ must be closed under the operation $d_{F_7}$ as desired.  As the other fluxes do not enter in this expression, one in fact finds that the complex $(H_7,F_5,F_3,F_1)$ is $d_{F_7}$ closed.  It was shown in Ref. \cite{monodromie} that $d_{F_7}$-exact fluxes are related by automorphisms generated by certain monodromies and so the quadruplet $(H_7,F_5,F_3,F_1)$ is classified by $\H_{F_7}^{\rm odd}(X^9)$.  

Recall from Eq.~(\ref{em}) that each p-form on $Y^{10}$ can be decomposed into a $p$-form on $X^9$ and a $(p-1)$-form on $X^9$ with one leg along the time direction.  These are called magnetic and electric fluxes respectively.  The above argument was given for the magnetic quadruplet $(H_7,F_5,F_3,F_1)$, but in fact, as we are using a magnetic $F_7$, it applies identically to the corresponding electric quadruplet which, being one degree lower, will then be classified by $\H_{F_7}^{\rm even}(X^9)$.  We therefore learn that spherical T-duality provides a one to one correspondence between the allowed electric and magnetic fluxes which is inequivalent to the familiar electromagnetic duality given by Hodge duality.

Now $dH_7$ is Poincar\'e dual to an F1-brane and $dF_1$ to a D7-brane and $dF_7$ to a D1-brane.  An application of $d$ yields
\beq
ddH_7=dF_1\wedge F_7+ F_1\wedge dF_7.
\eeq
The last term vanishes if we consider a configuration with no D1-branes.  Note that $ddH_7$ is nontrivial if there are F1-branes with boundaries, in which case it is Poincar\'e dual to the boundary.  Taking the cap product of this expression with the top class one finds that the boundary of the F1-branes is equal to integral of $F_7$ over the D7-branes, thus demonstrating the existence of the baryon configuration invoked in the previous subsection.

\section{Towards Geometry}

Thus far we have mainly considered spherical T-duality from a topological perspective.  However,
to determine whether this duality is actually a symmetry of some physical theory, we need to introduce some
geometry and determine a set of transformation rules analogous to the Buscher rules for ordinary T-duality  \cite{Buscher, Buscher2}.  

In order to get some insight into how spherical T-duality acts on concrete geometries, we discuss some explicit
examples of metrics, connections and 7-forms on certain principal $\sfSU(2)$-bundles.  
Concretely, we aim to construct a canonical metric on a principal $\sfSU(2)$-bundle $\pi:P\to M$ with 
2nd Chern number $c_2(P)$  of the form
\begin{equation} \label{eqn:0501a}
ds^2_P = ds^2_M + A \odot A\,,
\end{equation}
where $A$ is a principal connection on $P$, such that 
\begin{equation}
c_2(P) = \frac{1}{8\pi^2} \int_M \text{Tr} ( F_A \wedge F_A) \,.
\end{equation}

We will use a construction for base manifolds of the type $S^1\times M$ which involves 
a particular interpretation of the Chern-Simons form (see \cite{CS74}).  

To this end,  let $\pi:P\to M$ be a principal $\sfSU(2)$-bundle, and let
$A(t)$ be a path of principal connections on $P$ in $\Omega^1(P,\fg)$ such that $A(t)=A_0$, $A(1) = A_1$, then we have 
\begin{equation*}
\frac{1}{8\pi^2} \int_{[0,1]\times M} \text{Tr} (F_{A(t)} \wedge F_{A(t)}) = \int_{[0,1]\times M} d\CS(A(t)) = 
\cs(A_1) - \cs(A_0) \,,
\end{equation*}
where we have defined
\begin{equation*}
\cs(A) = \int_M \CS(A)  = \frac{1}{8\pi^2} \int_{M}   \text{Tr} ( A\wedge F - \frac13 A\wedge A\wedge A)   \,.
\end{equation*}
In particular, if we take $g:M\to \sfG$, $A_0=A$, $A_1={}^gA$, then 
\begin{equation*}
\int_{[0,1]\times M} \text{Tr} (F_{A(t)} \wedge F_{A(t)}) = \cs({}^gA) - \cs(A) = \text{deg}\ g\,.
\end{equation*}
Now, if we glue the endpoints of $[0,1]$ by the gauge transformation $g$, we obtain a $\sfG$-bundle $\widetilde P$
over $S^1\times M$ with $c_2 (\widetilde P)=  \text{deg}\ g$.  It thus remains to provide an explicit formula for $g:M\to \sfG$,
then we can construct a principal connection $A$ on $\widetilde P$, and the canonical metric 
\begin{equation} \label{eqn:0501b}
ds_{\widetilde P}^2 = ds^2_{S^1} + ds^2_M + A \odot A
\end{equation}
on the $\sfG$-bundle $\widetilde P$ over $S^1\times M$.
\medskip

Now take $\sfG=\sfSU(2)$, and $M=S^3$.  Coordinates on the $S^3$ base, and $S^3$-fiber,
are given in terms of unit quaternions $p$ and $q$, respectively.
The (left-invariant) Maurer-Cartan forms are given by $\bar p dp$ and $\bar q dq$, and a metric
for the trivial $\sfSU(2)$-bundle over $S^3$ is given by 
\begin{equation}
ds_{S^3\times S^3}^2 =  |\bar pdp|^2 + |\bar qdq|^2\,,
\end{equation}
which is of the form \eqref{eqn:0501a} with trivial principal connection $A = \bar qdq$.
A map $g:S^3\to \sfSU(2)$ of degree $k$ is given in terms of quaternions by $p \mapsto p^k$, and
acts on the $\sfSU(2)$-fiber as $q\to p^k q$,
so we have a principal connection on the $\sfSU(2)$-bundle $P$ over $S^1\times S^3$ with $c_2(P)=k$
by
\begin{equation*}
A = (1-t) \bar q dq + t (\overline{p^k q}) d(p^k q)  = \bar q dq +  t  \bar q (\bar p^k dp^k) q \equiv \bar q \, (d + \widetilde A) \, q \,.
\end{equation*}
[Note that $\bar p^k dp^k \neq k \, \bar p dp$ as $p$ and $dp$ do not commute.]

So, for $k=1$ we have $\widetilde A = t  \bar pdp$, and 
\begin{equation*}
\widetilde F = dt \wedge \bar p dp +  (t^2-t)\bar p dp \wedge \bar p dp 
\end{equation*}
Thus
\begin{equation*}
\widetilde F \wedge \widetilde F = 2 t (t-1) dt \wedge (\bar pdp)^3 \,,
\end{equation*}
and 
\begin{equation*}
c_2(P) = \frac{1}{8\pi^2} \int_{S^1\times S^3} \text{Tr} (\widetilde F \wedge \widetilde F) = 
-3 \ \cs(\bar p dp) \int_0^1 dt\  2 t (t-1)= 1 \,,
\end{equation*}
where we have used the normalization 
\begin{equation*}
\cs( \bar p dp)  = -\frac{1}{24\pi^2} \int_{S^3} \text{Tr} (\bar p dp)^3 = 1 \,.
\end{equation*}

Now, take $k=2$.  Then by the same arguments as before, we have 
\begin{equation*}
c_2(P) = \cs( \bar p^2 dp^2) =2\,,
\end{equation*}
where the result follows from our abstract arguments above.
On the other hand
\begin{equation*}
\bar p^2 dp^2 = \bar p dp + \bar p ( \bar p dp ) p \,.
\end{equation*}
We have 
\begin{equation*}
\cs( \bar p^2 dp^2) = \cs(\bar p dp + \bar p ( \bar p dp ) p)  = - \frac{1}{24\pi^2} \int_{S^3}  \text{Tr} (\bar p dp + \bar p ( \bar p dp ) p) ^3 \,.
\end{equation*}
Since 
\begin{equation*}
- \frac{1}{24\pi^2} \int_{S^3}  \text{Tr} (\bar p dp)^3 = - \frac{1}{24\pi^2} \int_{S^3}  \text{Tr}( \bar p ( \bar p dp ) p) ^3 =1 \,,
\end{equation*}
it follows that 
\begin{equation*}
\int_{S^3}  \text{Tr}( (\bar p dp)\wedge (\bar p dp) \wedge \bar p (\bar p dp) p ) = 
\int_{S^3}  \text{Tr}( (\bar p dp)\wedge  \bar p (\bar p dp) p \wedge \bar p (\bar p dp) p) = 0\,.
\end{equation*}
Similarly, for higher $k$, 
\begin{equation*}
\cs( \bar p^k dp^k) = \cs(\bar p dp + \bar p ( \bar p dp ) p + \ldots + \bar p^{k-1}  ( \bar p dp )  p^{k-1} ) \,,
\end{equation*}
and since 
\begin{equation*}
- \frac{1}{24\pi^2} \int_{S^3}  \text{Tr} (\bar p dp)^3 = \ldots = - \frac{1}{24\pi^2} \int_{S^3}  \text{Tr}( \bar p^{k-1} ( \bar p dp ) p^{k-1})^3 =1 \,,
\end{equation*}
it follows that all the mixed terms vanish.

To summarize, we have constructed an explicit metric and connection  on the $\sfSU(2)$-bundle $\pi : P\to S^1\times S^3$ of 2nd Chern 
number $c_2(P) = k \in H^4(S^1\times S^3,\ZZ) \cong \ZZ$ given by
\begin{align}
ds^2_P & = dt^2 + |\bar pdp|^2 + |\bar q dq +  t  \bar q (\bar p^k dp^k) q |^2 \,, \nonumber\\
A & = \bar q dq +  t  \bar q (\bar p^k dp^k) q \,.
\end{align}
\medskip

Now consider $\sfSU(2)$-bundles $P$ over $S^4$.  They are again completely classified by their 2nd Chern class $c_2\in H^4(S^4;\ZZ) \cong \ZZ$.  
The construction above does not quite work for $S^4$, but 
we can view the base $S^4$ as a smashed product $S^1 \wedge S^3$, and guess the metric and connection on $P$.  
Fortunately, we have an explicit description for $c_2=1$ where  $P=S^7$.  We view $S^7\subset \HH^2$, and choose coordinates
(resembling the Euler angles for $S^3$)
\begin{align*}
q_1 & = \cos \theta \,q \,,\\
q_2 & = \sin \theta\ p\, q \,.
\end{align*}
where $q,p$ are unit quaternions, i.e.\ can be identified with points in $S^3$, and $0\leq \theta \leq \frac{\pi}{2}$.  
Now, a short computation gives
\begin{align*}
ds^2_{S^7} = |dq_1|^2 + |dq_2|^2 & = d\theta^2 + \frac14 \sin^2 2\theta \ |\bar p dp|^2 + |\bar q dq + \sin^2 \theta\  \bar q (\bar pdp) q|^2 = 
ds^2_{S^4} + A \odot A  \,,
\end{align*}
where 
\begin{equation*}
A = \overline  q\, dq  +  \bar q ( \sin^2\theta\ \bar p dp) q  = \bar q(d+\widetilde A) q \,,
\end{equation*}
which is similar to our result for $S^1\times S^3$ with $t$ replaced by $\sin^2\theta$.  Similarly we have 
\begin{equation*}
\widetilde F = \sin2\theta d\theta \wedge (\bar p dp) - \frac14 \sin^2 2\theta (\bar p dp) \wedge (\bar p dp)\,,
\end{equation*}
and 
\begin{equation*}
\widetilde F \wedge \widetilde F = - \frac12 \sin^3 2\theta \  dt \wedge (\bar p dp)^3 \,.
\end{equation*}
Hence
\begin{equation*}
c_2(S^7) = \frac{1}{8\pi^2} \int_{S^4} \text{Tr} (\widetilde F \wedge \widetilde F) = 
3 \ \cs(\bar p dp) \int_0^{\pi/2} d\theta  \ \frac12 \sin^3 2\theta  = 1 \,,
\end{equation*}
where we have used  $\int_0^{\pi/2} d\theta  \  \sin^3 2\theta = \frac23$.
A natural guess for the metric on the principal 
$\sfSU(2)$ bundle $P$ over $S^4$ with $c_2(P)=k$ is thus given by 
\begin{equation*}
ds^2_P = ds^2_{S^4} + A \odot A
\end{equation*}
with 
\begin{equation*}
A = \overline  q\, dq  +  \bar q \, ( \sin^2\theta\ \bar p^k dp^k) \, q \,,
\end{equation*}
and an explicit calculation, using the results above, indeed shows that $c_2(P)=k$.
\medskip

Now we consider an explicit representative of a class in $H^7(P,\ZZ)$ for the principal $\sfSU(2)$-bundles we have 
just constructed.  
Consider $M=S^1\times S^3$, and let  $H\in \Omega^7_{\text{cl}}(P)$ be given by
\begin{equation*}
H = dt \wedge \CS(\bar p dp) \wedge \CS(A) = dt \wedge \CS(\bar p dp) \wedge \CS(\bar qdq) \,.
\end{equation*}
We see that this is a globally defined form on the $\sfSU(2)$-bundle $P$ over $S^1\times S^3$ with $c_2(P)=k$.  
Since 
\begin{equation*} 
\int_{P} dt \wedge \CS(\bar p dp) \wedge \CS(A) = \int_0^1 dt \ \times \cs(\bar p dp) \times\  \cs(\bar q dq) = 1 
\end{equation*}
it represents the generator of $H^7(P,\ZZ) \cong \ZZ$.
Now, locally, $H=dB$, with
\begin{equation*}
B = t \  \CS(\bar p dp) \wedge \CS(A) \,,
\end{equation*}
and 
\begin{equation*}
\pi_* B = t \  \CS(\bar p dp) \,.
\end{equation*}
So at the level of forms, spherical T-duality is the statement that
\begin{align*}
\widetilde A & = t (\bar p^k dp^k ) & \widetilde {\widehat{A}} & = t (\bar p^{k'} dp^{k'} ) \,,\\
B & = k' \,t \  \CS(\bar p dp) \wedge \CS(A) &  \widehat B & =  k\,  t \  \CS(\bar p dp) \wedge \CS(\widehat A) \,,
\end{align*}
which can be given the  interpretation of the exchange of winding number $k$ with the an `$S^3$-momentum' measured
by the legs of $B$ in the direction of the $S^3$-fiber, i.e.\ by $k' \ \CS(\bar q dq)$.
Similar formulas hold for $S^4$ with $t$'s replaced by appropriate functions of $\theta$.

Such an identification could be expected if, for example, spherical T-duality were a symmetry of the spectra of a theory of spherical 3-branes that can wrap 
$S^3$ cycles in some spacetime $X$, i.e.\ by replacing closed strings, described by $\text{Maps}(S^1,X)$, by spherical 3-branes
(or `closed quaternionic strings') described by $\text{Maps}(S^3,X) = \text{Maps}(S(\HH),X)$.

\section{Open questions and speculations}

In this Section we briefly list some open questions, speculations and directions for future research:

\begin{enumerate}

\item ($H$ a cocycle) 
In the case of T-duality of $\sfU(1)$-bundles, given the cohomology class represented by $H$ one 
could not uniquely determine that of $\widehat{H}$, but one could determine it up to a bundle automorphism.  
In the present case, 
again the cohomology class of $H$ does not determine that of $\widehat{H}$ and there is a 
gerbe automorphism which relates the 
choices of $\widehat{H}$, but sometimes that gerbe automorphism does not lift to a bundle 
automorphism and so there really are inequivalent choices of $\hat{H}$.  However it appears 
that if a specific cocycle $H$ is chosen and if one demands that $\widehat p^*H=p^*\widehat{H}$ 
as a cocycle, then $\widehat{H}$ will be completely determined as a cocycle and so 
also as a cohomology class.  Thus in this context there is no ambiguity in the determination 
of $\widehat{H}$.  Similarly, if $H$ is a gerbe with connection, as it is in string theory, and if 
one demands that the pullbacks of $H$ and $\widehat{H}$ to the correspondence space 
agree as gerbes with connection, then $\widehat{H}$ appears to be determined as a 1-gerbe 
with connection.  This comment applies both to spherical T-duality and to ordinary T-duality, 
and so will be treated separately elsewhere.

\item  (Missing spherical T-duals via noncommutative geometry?) When $dim(M)>4$, then 
given a pair $(P, H)$ over $M$ consisting of a principal
$\sfSU(2)$-bundle $P \to M$ together with a class $H \in \H^7(P,\ZZ)$, it does not 
in general have a spherical T-dual in the sense of the paper. The question is, what
can be said about these missing spherical T-duals? For instance, is it possible that
$(P, H)$ has a noncommutative spherical T-dual? 

A naive approach to this question is to use the generalized Dixmier-Douady theory in \cite{DP},
where it is shown that there is a an algebra bundle $\cO \to P$ over $P$ with fiber the stabilized Cuntz algebra
$\cO_\infty \otimes \cK$, where $\cK$ is the algebra of compact operators, and with (generalized)
Dixmier-Douady class ${\rm DD}(\cO) = H$. Suppose that the $\sfSU(2)$-action on $P$ lifts 
to an $\sfSU(2)$-action on the algebra of sections vanishing at infinity,  $C_0(P, \cO)$, and consider the crossed product algebra, $C_0(P, \cO)\rtimes \sfSU(2)$.
This could potentially be the missing spherical T-dual of $(P, H)$ as it has a coaction of $\sfSU(2)$
such that $(C_0(P, \cO)\rtimes \sfSU(2))\rtimes \sfSU(2) \cong C_0(P, \cO),$ by non-abelian Takai duality
\cite{Raeburn}. There are a couple of problems. The first is the lifting of the $\sfSU(2)$-action. It is likely that
it lifts, as a similar problem for 3-cocycles was always shown to be true in \cite{AS}. The second is a serious 
problem and shows why this naive approach is doomed, namely  the putative spherical T-dual
$C_0(P, \cO)\rtimes \sfSU(2)$ is not K-theory equivalent to $C_0(P, \cO)$ even with a degree shift.

\item (The higher rank case?) By the higher rank case we mean the following.  Consider pairs $(P, H)$ over $M$ consisting of a 
principal $\sfSU(2)^r$-bundles $P$ with flux $H\in \H^{7}(P, \ZZ)$. Alternatively, we might think of
the associated bundles $P\times_{(S^3)^r} \HH^r$ as quaternionic vector bundles.
When $r\geq1$ a generic flux $H\in \H^{7}(P, \ZZ)$, under dimensional reduction, will have a component in $\H^1(M)$.
This will be an obstruction to the existence of a `classical spherical T-dual'.  One may speculate that this $\H^1(M)$
will play a role as a noncommutativity parameter in some noncommutative spherical T-dual, in the same way 
as for T-duals of higher rank torus bundles with nonclassical H-fluxes.

\item (Higher twisted algebroids?) 
Generalized geometry provides a natural framework in which to discuss T-duality for principal $S^1$-bundles $P\to M$
(see, e.g., \cite{cavalcanti}).  Specifically, T-duality provides an isomorphism of Courant brackets as well as many other structures,
between the $S^1$-invariant parts of the ($H_3$-twisted) generalized tangent spaces 
$E = (TP \oplus T^*P)^{\text{inv}}$ and its T-dual 
$\widehat E = (T\widehat P \oplus T^*\widehat P)^{\text{inv}}$.  It is a natural question to ask whether there exist algebroids over
a manifold $P$ that can twisted by $H \in \Omega_{\text{cl}}^7(P,\RR)$, i.e.\ closed 7-forms, and exhibit spherical T-duality in the 
case $P$ is a principal $\sfSU(2)$-bundle over some base $M$.  A minimal candidate might be the Leibniz algebroid
$E=TP\oplus \wedge^5 T^*P$, but we have been unable to find a spherical T-dual in this case.  An alternative, suggested 
by many of the constructions in this paper, might be some kind of quaternionization of the standard Courant algebroid, e.g.\
$(TP \oplus T^*P)\otimes_\RR \HH$.  We leave this for further investigation.

\end{enumerate}


\end{document}